\documentclass[journal,onecolumn,twoside]{IEEEtran}
\normalsize
\ifCLASSINFOpdf
\else
\fi
\usepackage{setspace}
\usepackage{color}
\usepackage{cite}
\usepackage{amsmath}
\usepackage{amsfonts}
\usepackage{amssymb}
\usepackage{amsthm}
\usepackage{tikz}
\usepackage{cases}
\usepackage{bm}
\usepackage{graphicx}
    \graphicspath{{../}}
    \DeclareGraphicsExtensions{.pdf}
\usepackage[caption=true,font=footnotesize]{subfig}
\usepackage{multirow}
\usepackage{makecell}
\usepackage{mathdots}
\usepackage{booktabs}
\usepackage{url}
\usepackage{bm}
\usepackage{xtab}
\usepackage{pifont}
\usepackage{array}
\usepackage{algorithm}
\usepackage{algorithmic}
\usepackage{enumerate}
\usetikzlibrary{arrows}
\usepackage{caption}

\usepackage{lineno} 

\usepackage{mathrsfs}
\usepackage{tikz-cd}
\usepackage{mathtools}
\usepackage{makecell}
\usepackage{booktabs}
\usepackage{array}
\usepackage{enumitem}

\allowdisplaybreaks[4]
\usepackage[colorlinks,
            linkcolor=blue,
            anchorcolor=blue,
            citecolor=blue,
            urlcolor=black]{hyperref}
\theoremstyle{plain}

\newtheorem{theorem}{Theorem}[section]
\newtheorem{lemma}{Lemma}[section]

\newtheorem{definition}{Definition} [section]
\newtheorem{corollary}{Corollary} [section]
\theoremstyle{definition}
\newtheorem{example}{Example}[section]

\newtheorem{remark}{Remark}[section]

\begin{document}	
\title{Bounds and Optimal Constructions of Generalized Merge-Convertible Codes for Code Conversion into LRCs}
\author{Haoming Shi,~\IEEEmembership{}
        Weijun~Fang,~\IEEEmembership{}
        Yuan~Gao~\IEEEmembership{}
\IEEEcompsocitemizethanks{\IEEEcompsocthanksitem Haoming Shi, Weijun Fang and Yuan Gao are with State Key Laboratory of Cryptography and Digital Economy Security, Shandong University, Qingdao, 266237, China, Key Laboratory of Cryptologic Technology and Information Security, Ministry of Education, Shandong University, Qingdao, 266237, China and School of Cyber Science and Technology, Shandong University, Qingdao, 266237, China (emails: 202421328@mail.sdu.edu.cn, fwj@sdu.edu.cn, gaoyuan862023@163.com).
}
\thanks{The work is supported in part
by the National Key Research and Development Program of China under Grant Nos. 2022YFA1004900 and 2021YFA1001000, the National Natural Science Foundation of China under Grant No. 62201322, the Natural Science Foundation of Shandong Province under
Grant No. ZR2022QA031, and the Taishan Scholar Program of Shandong Province. {\it (Corresponding Author: Weijun Fang)}.}
\thanks{Manuscript submitted }}

\maketitle
\begin{abstract}
Error-correcting codes are essential for ensuring fault tolerance in modern distributed data storage systems. However, in practice, factors such as the failure rates of storage devices can vary significantly over time, resulting in changes to the optimal code parameters. To reduce storage cost while maintaining efficiency, Maturana and Rashmi introduced a theoretical framework known as code conversion, which enables dynamic adjustment of code parameters according to device performance. In this paper, we focus exclusively on the bounds and constructions of generalized merge-convertible codes. First, we establish a new lower bound on the access cost when the final code is an $(r,\delta)$-LRC. This bound unifies and generalizes all previously known bounds for merge conversion, where the initial and final codes are either LRCs or MDS codes.  We then construct a family of access-optimal  MDS convertible codes by leveraging subgroups of the automorphism group of a rational function field. It is worth noting that our construction is also per-symbol read access-optimal. Next, we further extend our MDS-based construction to design access-optimal convertible codes for the conversion between $(r,\delta)$-LRCs with parameters that have not been previously reported. Finally, using the parity-check matrix approach, we present a construction of access-optimal convertible codes that enable merge conversion from MDS codes to an $(r,\delta)$-LRC. To the best of our knowledge, this is the first explicit optimal construction of code conversion between MDS codes and LRCs. All of our constructions are performed over finite fields whose sizes grow linearly with the code length.

\end{abstract}


\IEEEpeerreviewmaketitle

\section{Introduction}
\label{se1}

Erasure codes have gained significant attention in recent years for their application in distributed storage systems, where they protect against node failures with minimal overhead \cite{sathiamoorthy2013xoring}. In these systems, \( k \) data symbols are encoded into \( n \) codeword symbols and distributed across \( n \) distinct storage nodes. The code rate, denoted as \( \frac{k}{n} \), is selected based on the failure rate of the storage devices. In practical settings, the failure rate of storage devices in large-scale distributed systems can vary significantly over time. However, the erasure code rate typically remains fixed, which can lead to either excessive resource consumption or increased risk as time progresses. To mitigate this issue, \cite{kadekodi2019cluster} demonstrated that dynamically adjusting the code rate in response to fluctuating failure rates results in considerable savings in both storage space and operational costs. To adjust the code rate, one needs to modify the parameters of the code. This process, formally introduced in \cite{maturana2022convertible}, is referred to as code conversion.

\subsection{Related work}

Realizing code conversion is essentially equivalent to constructing a convertible code (see Definition \ref{concod}). Maturana and Rashmi initiated the theoretical study of convertible codes in \cite{maturana2022convertible}, focusing on merge conversion from MDS codes to an MDS code and they used access cost as a metric for conversion efficiency. They derived a lower bound on the access cost when both the initial codes and the final code are MDS codes, and presented two constructions that are optimal with respect to this bound.

In a follow-up study, Maturana et al. \cite{maturana2020access} extended the framework of MDS convertible codes to a more general set of parameters, including the split regime. Later, Maturana and Rashmi investigated the conversion problem for MDS array codes in \cite{maturana2023bandwidth}, where they established lower bounds on bandwidth cost for both merge and split scenarios. By leveraging the piggybacking framework, they also proposed explicit constructions of MDS convertible array codes that achieve optimal bandwidth cost with respect to these bounds. 

Other related work includes \cite{rashmi2017piggybacking} and \cite{mousavi2018delayed}, where researchers explored bandwidth-efficient schemes by introducing extra parities to MDS array codes. In the language of convertible codes, they considered the regime where \( t = 1 \) and \( n_I < n_F \). In another direction,  Su et al. explored the code conversion of polynomial codes in the context of distributed matrix multiplication in \cite{su2022local}. Meanwhile, Wu et al. proposed a redistribution-based method for converting RS codes with low access cost in \cite{wu2020enabling}. Recently, Ge et al. \cite{ge2024mds} generalized the definition of convertible codes by allowing each initial and final code to have different parameters. They analyzed the lower bounds on access cost for both merge and split conversions when all the involved codes are MDS codes and further provided constructions that achieve these bounds.

 To address dynamic workload patterns in real-world storage systems, Xia et al. \cite{xia2015tale} proposed a storage architecture that employs two different LRCs to separately encode frequently accessed data and less frequently accessed data. A key component of their system is a conversion mechanism between the two LRCs with different information-symbol locality parameters, which reduces the number of nodes accessed during conversion, thereby minimizing the read access cost. More recently, Maturana and Rashmi \cite{maturana2023locally} investigated the code conversion problem for locally repairable array codes with information locality. By using the piggybacking techniques, they introduced a construction that enables more efficient code conversion (in terms of bandwidth cost) compared to the conventional approach. In addition, Kong \cite{kong2024locally} introduced a construction of LRC convertible codes with optimal access cost based on the Tamo-Barg code.

Nonetheless, to the best of our knowledge, prior work on read access cost has typically focused on the total number of reads, while the per-symbol read access cost (see the definition in Remark \ref{perred}) has not been extensively studied. Additionally, research on the conversion between LRCs is limited, especially for variants such as $(r,\delta)$-LRCs, whose conversion has not been addressed in the literature (but see the Remark \ref{remark49}). Furthermore, previous studies have only considered conversions within the same class of codes, whereas the conversion from MDS codes to $(r,\delta)$-LRCs has not been explored.

\subsection{Main Contributions and Comparison}
In this paper, we focus exclusively on constructing access-optimal generalized merge-convertible codes. Below, we summarize the main contributions of this paper.
\begin{itemize}
    \item First, we derive a new lower bound on the access cost when the final code is an $(r,\delta)$-LRC, which answers an open problem proposed in \cite{kong2024locally}. Our lower bound applies to a wide range of scenarios, and it subsumes all the currently known bounds as special cases.
    \item Second, by leveraging the structure of the automorphism group of a rational function field and examining the action of carefully selected automorphisms on its rational places, we propose a new construction of access-optimal convertible codes from MDS codes to an MDS code, which is also per-symbol read access-optimal.  
    \item Third, by selecting two nested subgroups $\mathcal{H} \subset \mathcal{G}$ within the automorphism group of a rational function field, we use $\mathcal{H}$ to ensure that the constructed code satisfies the $(r, \delta)$ local property. Then, by carefully choosing elements from $\mathcal{G} \setminus \mathcal{H}$ and applying them to specific rational places of the function field, we are able to construct convertible codes that merge $(r,\delta)$-LRCs into an $(r,\delta)$-LRC. Such a construction is access-optimal with respect to the bound we established in Section \ref{sec3}.
    \item Finally, by employing GRS codes and the parity-check matrix approach to construct optimal $(r, \delta)$-LRCs, we present a construction of an access-optimal convertible code that converts MDS codes to an optimal $(r,\delta)$-LRC. To the best of our knowledge, this is the first explicit optimal construction of code conversion between MDS codes and LRCs.
\end{itemize}

All of our constructions are performed over finite fields whose sizes grow linearly with the code length. Moreover, the range of our parameters is quite broad, making our constructions highly applicable to practical storage systems.

We compare our main results with the existing related results in Table \ref{tab:table1} and Table \ref{tab:table2}.
Since we only consider merge regime, in the table, we use code $\mathcal{C}^{I_i}$ with parameter $[n_{I_i},k_{I_i}]$ to denote the initial codes for $i \in [t]$,  and use code $\mathcal{C}^{F}$ with parameter $[n_{F},k_{F}]$ to denote the final code, we also use $l_{I_i} \triangleq n_{I_{i}}-k_{I_i} $ for $i \in [t]$ and $l_{F} \triangleq n_{F}-k_{F}$ to denote the redundancy of these codes.

\renewcommand{\arraystretch}{3}
\begin{table}[tbhp] 
\scriptsize
\caption{LOWER BOUNDS ON ACCESS COST FOR GENERALIZED MERGE-CONVERTIBLE CODE}
\label{tab:table1}
\begin{center}
    \begin{tabular}{|c|c|c|c|p{2.7cm}|} \hline
        \textbf{Condition On Codes} & \textbf{ Read Access Cost} & \textbf{ Write Access Cost}& \textbf{ Reference} \\ \hline
        
        \makecell{$\mathcal{C}^{I_1},\cdots, \mathcal{C}^{I_t} \text{~and~} \mathcal{C}^{F}$ \\ $\text{~are MDS codes }$} & \makecell{$\sum\limits_{i\in [t],l_F \leq \text{~min~} \{k_{I_i},l_{I_i}\}}l_F $ \\ + $\sum\limits_{i\in [t],l_F > \text{~min~} \{k_{I_i},l_{I_i}\}}k_{I_i} $} & $l_{F}$ & \cite[Theorem 1]{ge2024mds} \\ \hline
         \makecell{$\mathcal{C}^{I_1}=\cdots =\mathcal{C}^{I_t}$, \\ $\mathcal{C}^{F}\text{~is an $r$-LRC }$} &   $\makecell{ \begin{cases}
             k_F \text{,~if~} \Delta \leq 0 \text{~or~} d_F > n_{ I_1} - k_{I_1} +1; \\ k_F-t\lceil \frac{r\Delta}{r+1} \rceil, \text{~otherwise.}
         \end{cases}  \\ \text{~where~} \Delta = k_{I_1}+  n_F-k_F \\ -2d_F+3- \lceil \frac{(t-1)k_{I_1}}{r} \rceil  }$   &  $ \makecell{-(t-1)n_F+(t-1)k_F\\+td_F-2t+  t \lceil \frac{(t-1)k_{I_1}}{r} \rceil  } $   & \cite[Theorem IV.4]{kong2024locally} \\ \hline
        $\mathcal{C}^{F}\text{~is an $(r,\delta)$-LRC  }$  & $\makecell{ k_F - \sum_{ \substack{i \in [t], ~\widetilde{\Delta_i} >0 \\ and~ d_F \leq n_{I_{i}}-k_{I_i}+1} }  [\widetilde{\Delta_i} - (\delta-1) \lfloor \frac{\widetilde{\Delta_i}}{r+\delta-1}  \rfloor] \\ \text{~where~} \widetilde{\Delta_i} = k_{I_i}+  n_F-k_F-2d_F \\ +2-( \lceil \frac{k_F-k_{I_i}}{r} \rceil -1)(\delta-1) }$ & $ \makecell{-(t-1)n_F+(t-1)k_F+td_F\\-t+ \sum_{i=1}^{t}( \lceil \frac{k_F-k_{I_i}}{r} \rceil -1)(\delta-1) } $ & Theorem \ref{bouthm} \\ \hline
        
         \makecell{$\mathcal{C}^{F} \text{~is an}$ $(r,\delta)$-LRC   with \\ $k_F=\sum\limits_{i=1}^{t}k_ir$  and $k_{I_i}=k_ir$ } & $k_F - \sum_{ \substack{i \in [t], ~-\frac{n_F}{r
          +\delta-1} + k_i+ \frac{k_F}{r} >0 \\~and ~ d_F \leq n_{I_{i}}-k_{I_i}+1} }    [k_F+ rk_i-\frac{rn_F}{r+\delta-1}] $ & $n_F-(r+\delta-1)\sum\limits_{i=1}^{t}k_i$ & Corollary \ref{rdelcor} \\ \hline
        
    \end{tabular}
\end{center}
\end{table}
\renewcommand{\arraystretch}{1.0}

\begin{table}[tbhp]
\scriptsize
\renewcommand{\arraystretch}{2.5}  
\caption{SOME KNOWN CONSTRUCTIONS OF ACCESS-OPTIMAL GENERALIZED MERGE-CONVERTIBLE CODES}
\label{tab:table2}
\begin{center}
\begin{tabular}{|c|p{3.6cm}|p{4.6cm}|p{3.1cm}|}
\hline
\textbf{Construction and code type} & \textbf{Restrictions} & \textbf{Field size $q$ requirement} & \textbf{Reference} \\ \hline

\makecell{General \\ all codes are MDS codes} &
$C^{I_1} = C^{I_2} = \cdots = C^{I_t}$, \newline $l_F \leq \min\{k_{I_1}, l_{I_1}\}$ &
$q \geq \max\left\{2^{\mathcal{O}\left((n_F)^3\right)},\ n_{I_1} - 1\right\}$ &
\cite[Theorem 21]{maturana2022convertible} \\ \hline

\makecell{Hankel-I \\ all codes are MDS codes} &
$C^{I_1} = C^{I_2} = \cdots = C^{I_t}$, \newline $l_F \leq \lfloor l_{I_1}/t \rfloor$ &
$q \geq \max\{n_F, n_{I_1}\} - 1$ &
\cite[Example 22]{maturana2022convertible}  \\ \hline

\makecell{Hankel-II \\ all codes are MDS codes} &
$C^{I_1} = C^{I_2} = \cdots = C^{I_t}$, \newline $l_F \leq l_{I_1} - t + 1$ &
$q \geq \max\{k_{I_1}l_{I_1}, n_{I_1} - 1\}$ &
\cite[Example 23]{maturana2022convertible} \\ \hline

\makecell{ $\mathrm{Hankel}_s$, for $t \leq s \leq l_{I_1}$ \\ all codes are MDS codes} &
$C^{I_1} = C^{I_2} = \cdots = C^{I_t}$, \newline $l_F \leq (s-t+1)\lfloor l_{I_1}/s \rfloor + \max\{(l_{I_1} \bmod s) - t + 1, 0\}$ &
$q \geq \max\left\{sk_{I_1} + \lfloor r_{I_1}/s \rfloor -1,\ n_{I_1} - 1\right\}$ &
\cite[Theorem 24]{maturana2022convertible} \\ \hline

\makecell{GRS \\ all codes are MDS codes} &
$C^{I_1} = C^{I_2} = \cdots = C^{I_t}$,  \newline  $l_F \leq \min\{k_{I_1}, l_{I_1}\}$ & $ \max\{k_{I_1}, l_{I_1}\} |(q-1)$,  \newline 
$q \geq (t+1)\max\{k_{I_1}, l_{I_1}\} + 1$ &
\cite[Corollary II.2]{kong2024locally} \\ \hline

\makecell{Extended GRS \\ all codes are MDS codes} &
None &
$q \geq \max\{n_{I_1}, \cdots, n_{I_t}, n_F\} - 1$ &
\cite[Theorem 4]{ge2024mds} \\ \hline

\makecell{Tamo-Barg codes \\ all codes are optimal LRCs} &
$C^{I_1} = C^{I_2} = \cdots = C^{I_t}$, \newline $l_F \leq \text{min~} \{ k_{I_1},l_{I_1} \}, r|k_{I_1}$ &
$ \frac{k_{I_1}}{r}(r+1)  |(q-1) $, \newline $ q \geq \frac{k_{I_1}}{r}(r+1) \text{max}\{ t+1, \lceil \frac{l_{I_1}r}{k_{I_1}}\rceil +2 \}+1$ &
\cite[Corollary III.2]{kong2024locally} \\ \hline

\makecell{Construction in section \ref{sec4} \\ all codes are MDS codes} &
 $ t \leq l_F $ &
 $\exists$ a subgroup of order \( l_F  \) in  $\mathrm{PGL}_2(q)$, \newline  $ q \geq \text{max~}\{ n_{I_1},\cdots, n_{I_t},n_F\}-1 $ &
Theorem \ref{mdsthm} \\ \hline

\makecell{Construction in section \ref{sec5} \\ all codes are optimal $(r,\delta)$-LRCs} &
 $ t \leq l_F/(r+\delta-1
)$ \newline $r |k_{I_i}$ $\text{~for~} i \in [t]$ &
 $\exists$ nested subgroups of order \( l_F  \) and order \( r+\delta-1  \) in  $\mathrm{PGL}_2(q)$, \newline  $ q \geq \text{max~}\{ n_{I_1},\cdots, n_{I_t},n_F\}-1 $ &
Theorem \ref{rdelthm} \\ \hline

\makecell{Construction in section \ref{sec6} \\ MDS codes to an optimal $(r,\delta)$-LRC} &  $k_{I_1} = k_{I_2} = \cdots = k_{I_t}=k_I$, \newline $   k_{I} + 1 \leq \langle tk_I\rangle_r , d_F \leq k_{I}, d_F \leq n_{I_{i}}-k_{I} +1 $ for $ i \in[t]$
 & $q \geq n_{I_{i}}$ for $i \in [t]$, \newline $q \geq n_F-(t'-1)(\delta-1)$ for $t=st'$
  & Theorem \ref{thmVI.2}  \\ \hline

\end{tabular}
\end{center}
\end{table}

\renewcommand{\arraystretch}{1}

\subsection{Organization of the Paper}
In Section \ref{sec2}, we provide the necessary background for this paper, including the definitions of locally repairable codes, generalized merge-convertible codes, algebraic geometry codes, rational function field and its automorphism group. Section \ref{sec3} establishes a new lower bound on the access cost of generalized merge-convertible codes when the final code is an $(r,\delta)$-LRC. In Section \ref{sec4}, we propose a general method for constructing access-optimal MDS generalized merge-convertible codes from general subgroups of  $\mathrm{PGL}_2(q)$, and we provide explicit constructions based on cyclic subgroups of order dividing $q+1$. Section \ref{sec5} extends this framework to design new $(r,\delta)$ merge-convertible codes and demonstrates that our constructions attain the access cost lower bound established in Section \ref{sec3}. We also provide explicit examples using subgroups of the affine linear group and a dihedral group which can be viewed as a subgroup of $\mathrm{PGL}_2(q)$. In Section \ref{sec6}, we present the first construction of access-optimal generalized merge-convertible codes that convert MDS codes into an $(r,\delta)$-LRC. Finally, we conclude the paper by summarizing our main results in Section \ref{sec7}.

\section{Preliminaries}
\label{sec2}

\subsection{Locally Repairable Codes}
Throughout this paper, we use $[a,b]$ to denote the set $\{a,a+1,\cdots,b\}$ for any two integers $a \leq b$ and $[n]$ to represent the set $\{1,2,\cdots,n \}$.

The concept of locally repairable codes (LRCs) is essential for enhancing the reliability and efficiency of distributed storage systems. Informally, a block code $\mathcal{C} \subseteq \mathbb{F}_q^n$ is said to have locality $r$ if its value of any coordinate can be recovered by accessing at most $r$ other coordinates.  The formal definition of a locally repairable code with locality $r$ is given as follows.

\begin{definition}[Locally Repairable Codes,\cite{gopalan2012locality}] \label{lrcdf}
    A code $\mathcal{C}$ is called an $\mathrm{LRC}$ with locality $r$, if for every $i \in [n]$, there exists a subset $A_i \subseteq [n]\setminus \{i\}$ with $|A_i| \leq r$ and a function $\chi_i$, such that for each codeword $x \in \mathcal{C}$, the value $x_i$ is given by   \[ x_i = \chi_i(\{x_j : j \in A_i\}).  \]

This can also be expressed as follows. For any $a \in \mathbb{F}_q$, consider the set of codewords
\[
\mathcal{C}(i, a) = \{ x \in \mathcal{C} : x_i = a\}, i\in [n].
\]
The code $\mathcal{C}$ is said to have locality $r$ if for every $i$, there exists a subset $A_i \subseteq [n]\setminus \{i\}$ of cardinality at most $r$ such that the restrictions of  $\mathcal{C}(i, a)$ to the coordinates in $A_i$  are disjoint for different values of $a$, i.e.
\[
\mathcal{C}_{A_i}(i, a) \cap \mathcal{C}_{A_i}(i, a')=\emptyset, \forall a \neq a'.
\]
\end{definition}

As a natural generalization of the classical Singleton bound, \cite{papailiopoulos2014locally} established a bound on the minimum distance of an arbitrary $(n,k;r)$-LRC. We refer to the following version from \cite{tamo2014family} here.

\begin{theorem} [Singleton-type Bound, \cite{tamo2014family}] 
    Let $\mathcal{C}$ be an $(n,k,d; r)$-$\mathrm{LRC}$ of cardinality $|\Sigma|^k$ over an alphabet $\Sigma$, then the Singleton-type bound for an $\mathrm{LRC}$ is given by:
    \begin{equation}\label{Sbound1}
     d \leq n- k - \lceil \frac{k}{r} \rceil +2.
    \end{equation}
\end{theorem}

The concept of LRCs can be further generalized to accommodate multiple erasures.
\begin{definition}[$(r,\delta)$-Locally Repairable Codes,\cite{kamath2014codes}]\label{rdeldf}
    A code $\mathcal{C} \subseteq \mathbb{F}_q^n $ of size $q^k$ is said to have the $(r,\delta)$-locality where $\delta \geq 2$, if each coordinate $i \in [n]$ belongs to a subset $J_i \subset [n]$ of size at most $r+\delta-1$, such that the restriction $\mathcal{C}|_{J_i}$ to the coordinates in $J_i$ forms a code of minimum distance at least $\delta$.
\end{definition}

The subset $J_i$ in this definition is called the recovery set of $i$. For general $(r,\delta)$-LRCs, the Singleton-type bound is given as follows:

\begin{theorem} [Singleton-type Bound, \cite{kamath2014codes}]
    Let $\mathcal{C}$ be an $(n,k,d; (r, \delta))$-$\mathrm{LRC}$ of cardinality $|\Sigma|^k$ over alphabet $\Sigma$. Then, we have  
    \begin{equation} \label{Sbound2}
     d \leq n- k +1 - (\lceil \frac{k}{r} \rceil -1)(\delta -1).
    \end{equation}
\end{theorem}

When $\delta =2$, an $(r,\delta)$-LRC reduces to a standard LRC as in Definition \ref{lrcdf}. Since the minimum distance of $\mathcal{C}|_{J_i}$ is at least  $\delta$, any $\delta-1$ coordinates of $J_i$ are determined by the remaining $|J_i|-(\delta-1) \leq r$ coordinates, thereby enabling local recovery.

\subsection{Generalized Convertible Code}

Now we present the formal framework introduced in \cite{maturana2022convertible} and \cite{ge2024mds} for studying code conversions.

Suppose we have $t_1$ linear initial codes $\mathcal{C}^{I_1}$, $\mathcal{C}^{I_2}$, $\cdots$, $\mathcal{C}^{I_{t_1}}$ and $t_2$ linear final codes $\mathcal{C}^{F_1}$, $\mathcal{C}^{F_2}$, $\cdots$, $\mathcal{C}^{F_{t_2}}$, where for each $i\in[t_1]$, $\mathcal{C}^{I_i}$ is an $[n_{I_i}, k_{I_i}]_q$ code with redundancy $l_{I_i} \triangleq n_{I_{i}}-k_{I_i}$, and for each $j\in[t_2]$, $\mathcal{C}^{F_j}$ is an $[n_{F_j}, k_{F_j}]_q$ code with redundancy $l_{F_j} \triangleq n_{F_{j}}-k_{F_j}$. 

\begin{definition}
    Using the above notation, we say that the codes $\mathcal{C}^{F_1} , \cdots, \mathcal{C}^{F_{t_2}}$ are generated by the codes $\mathcal{C}^{I_1} , \cdots, \mathcal{C}^{I_{t_1}}$ if there exists a surjective map $\phi$: $ \mathcal{C}^{I_1} \times \cdots \times \mathcal{C}^{I_{t_1}} \longrightarrow \mathcal{C}^{F_1} \times \cdots \times \mathcal{C}^{F_{t_2}} $.
\end{definition}

Informally, in code conversion, we need to generate the final codes from the initial codes. Moreover, since both the initial and final codes should encode the same information, we assume that 
\begin{equation} \label{dimeq}
\sum_{i=1}^{t_1} k_{I_i} = \sum_{j=1}^{t_2} k_{F_j}.
\end{equation}

Mathematically, the definition of convertible codes is very simple.
\begin{definition}[Generalized Convertible Code] \label{concod}
    Let $t_1,t_2$ be two positive integers. A $(t_1,t_2)_q$ generalized convertible code $\mathscr{C}$ over $\mathbb{F}_q$ consists of:
  \begin{itemize}
      \item $t_1$ initial codes $\mathcal{C}^{I_1}$, $\mathcal{C}^{I_2}$, $\cdots$, $\mathcal{C}^{I_{t_1}}$ and $t_2$ final codes $\mathcal{C}^{F_1}$, $\mathcal{C}^{F_2}$, $\cdots$, $\mathcal{C}^{F_{t_2}}$ defined as above satisfying Eq. \eqref{dimeq}.
       \item A surjective map $\phi$ from the Cartesian product of the initial codes onto the Cartesian product of the final codes, i.e.,  \[\exists~\phi: \mathcal{C}^{I_1} \times \cdots \times \mathcal{C}^{I_{t_1}} \twoheadrightarrow \mathcal{C}^{F_1} \times \cdots \times \mathcal{C}^{F_{t_2}}.\]
  \end{itemize}
\end{definition}
\begin{remark}
    Since the dimensions of the spaces on both sides are equal, any surjective map between them is necessarily bijective.
\end{remark}

To better characterize the access cost (see Definition~\ref{acccost}) during the conversion process from a systems perspective, we provide a more detailed discussion below. For clarity, we assume that the symbols of the initial code $\mathcal{C}^{I_i}$ are indexed by $\{i\} \times [n_{I_i}]$, whereas those of the final code $\mathcal{C}^{F_j}$ are indexed by $\{t_1 + j\} \times [n_{F_j}]$.

Specifically, each symbol in any final code must be represented as a function of some symbols from the initial codes,  symbols in the initial code $\mathcal{C}^{I_i}$ for $i \in [t_1]$ can be classified into three categories based on their roles in the conversion process:
\begin{itemize}
    \item {\bf Retired symbols}:  Symbols in $\mathcal{C}^{I_i}$ that are not used in the conversion.
    \item {\bf Unchanged symbols} for $ j \in [t_2]$: Symbols in $\mathcal{C}^{I_i}$ that remain unchanged correspond to a certain symbol in $\mathcal{C}^{F_j}$. These symbols are indexed by $\mathcal{U}_{i,j} \subseteq \{i\}\times [n_{I_i}] $, and are denoted by $\mathcal{C}^{I_i}|_{\mathcal{U}_{i,j}}$.
    \item {\bf Read symbols} for $j \in [t_2] $: Symbols in $\mathcal{C}^{I_i}$ that participate in the conversion process to generate new symbols in $\mathcal{C}^{F_j}$, but are not merely unchanged symbols. They are indexed by $\mathcal{R}_{i,j} \subseteq \{i\}\times [n_{I_i}] $, and denoted by $\mathcal{C}^{I_i}|_{\mathcal{R}_{i,j}}$.
\end{itemize}

Similarly, symbols in the final code $\mathcal{C}^{F_j}$ for $j \in [t_2] $ fall into two categories:
\begin{itemize}
    \item {\bf Unchanged symbols} for $i \in [t_1]$:  Symbols in $\mathcal{C}^{F_j}$ that remain identical to their corresponding symbols in $\mathcal{C}^{I_i}$ are indexed by a subset of $\{t_1+j\} \times [n_{F_j}]$. However, to clarify their origin, we use the same index set $\mathcal{U}_{i,j} \subseteq \{i\} \times [n_{I_i}]$ to denote these unchanged symbols, which are written as $\mathcal{C}^{F_j}|_{\mathcal{U}_{i,j}}$. Clearly, we have $ \mathcal{C}^{I_i}|_{{ \mathcal{U}_{i,j}}} = \mathcal{C}^{F_j}|_{{ \mathcal{U}_{i,j}}}$.
    \item {\bf Written symbols}: Symbols in $\mathcal{C}^{F_j}$ that are not unchanged symbols from any initial code. These are indexed by $ \mathcal{W}_j \subseteq \{t_1+j\}\times [n_{F_j}]$ and written as $\mathcal{C}^{F_j}|_{\mathcal{W}_j}$. In code conversion, these symbols are generated by the read symbols in the initial codes.
\end{itemize}

In the storage systems, maintaining unchanged symbols has numerous practical benefits because such symbols can remain in the same physical location during the conversion procedure. This greatly improves the conversion efficiency.

Graphically, to simplify notation, we uniformly use $\operatorname{Res}$ to denote the obvious restriction map from the original code to a specified domain, then we have $t_2$ surjective maps $\phi_1, \phi_2, \cdots, \phi_{t_2}$ along with the identity map $\operatorname{Id}$ ensuring that the following diagrams commute:

\begin{equation}\label{diag1}
    \begin{tikzcd}
    \mathcal{C}^{I_1} \times \cdots \times \mathcal{C}^{I_{t_1}} \arrow[r, "\phi"] \arrow[d, "\mathrm{Res}"] & \mathcal{C}^{F_1} \times \cdots \times \mathcal{C}^{F_{t_2}} \arrow[d, "\mathrm{Res}"] \\
    \mathcal{C}^{I_1}|_{\mathcal{R}_{1,j}} \times \cdots \times \mathcal{C}^{I_{t_1}}|_{\mathcal{R}_{t_{1},j} } \arrow[r,twoheadrightarrow,"\phi_j"] & \mathcal{C}^{F_{j}}|_{\mathcal{W}_j}
\end{tikzcd}
\end{equation}
and
\begin{equation}\label{diag2}
    \begin{tikzcd}
    \mathcal{C}^{I_1} \times \cdots \times \mathcal{C}^{I_{t_1}} \arrow[r, "\phi"] \arrow[d, "\mathrm{Res}"] & \mathcal{C}^{F_1} \times \cdots \times \mathcal{C}^{F_{t_2}} \arrow[d, "\mathrm{Res}"] \\
    \mathcal{C}^{I_1}|_{\mathcal{U}_{1,j}} \times \cdots \times \mathcal{C}^{I_{t_1}}|_{\mathcal{U}_{t_{1},j} } \arrow[r,"\mathrm{Id}"] & \mathcal{C}^{F_{j}}|_{\bigcup\limits_{i \in [t_1]} \mathcal{U}_{i,j}}
\end{tikzcd}
\end{equation}
for $j \in [t_2]$.

Throughout this paper, we focus exclusively on {\bf linear} initial codes and the final codes with a {\bf linear} conversion map, which means that $\phi$ is a linear map. In particular, $\mathscr{C}$ is said to be an MDS (resp. $(r,\delta)$) generalized convertible code if all initial and final codes are MDS codes (resp. optimal $(r,\delta)$-LRCs).

Given a generalized convertible code, we use the access cost as a metric to evaluate the efficiency of the conversion process.

\begin{definition}[Access Cost,\cite{maturana2022convertible}]\label{acccost}
    Let $\mathscr{C}$ be a $(t_1,t_2)_q$ generalized convertible code. We define the read access cost as 
    \[\rho_r=\sum_{i \in [t_1] }|\cup_{j \in [t_2]} \mathcal{R}_{i,j} |,\] which represents the total number of symbols read during the conversion process.  Similarly, the write access cost is given by 
    \[\rho_w=\sum_{j \in [t_2] }(n_{F_j}-\sum_{i\in [t_1]}|\mathcal{U}_{i,j}|)\]
    which represents the total number of symbols written during the conversion process. The access cost of a generalized convertible code $\rho$ is the sum of these two values, i.e.,
    \[\rho=\rho_r+\rho_w.\]
    A convertible code is said to be access-optimal if it minimizes the total access cost among all convertible codes with the same parameters.
\end{definition}

Two categories of generalized convertible codes are of particular interest.
\begin{definition} [\cite{ge2024mds}]
    Let $\mathscr{C}$ be a $(t_1,t_2)_q$ generalized convertible code. Then \\
    1) $\mathscr{C}$ is called a generalized merge-convertible code if $t_1 > 1$ and $t_2 = 1$; \\
    2) $\mathscr{C}$ is called a generalized split-convertible code if $t_1 = 1$ and $t_2 > 1$.
\end{definition}

For the access cost, several results have already been established, we list them here.
\begin{theorem} [\cite{ge2024mds}] \label{mdsbou}
    Let $\mathscr{C}$ be a $(t,1)_q$ generalized merge-convertible code  in which all the initial codes and the final code are $\mathrm{MDS}$ codes, then the access cost satisfies
    \[
    \rho \geq l_{F_1} + \sum\limits_{i \in [t],~l_{F_{1}} \leq \text{~min~} \{ k_{I_i},l_{I_{i}} \} } l_{F_1} + \sum\limits_{i \in [t],~l_{F_{1}} > \text{~min~} \{ k_{I_i},l_{I_{i}} \} } k_{I_i}.
    \]
\end{theorem}

\begin{theorem} [\cite{kong2024locally}] \label{lrcbou}
    Let $\mathscr{C}$ be a $(t,1)_q$ generalized merge-convertible code in which all the initial codes have the same parameters $(n_I,k_I=k)$ and the final code is an $\mathrm{LRC}$ with locality $r$ and minimum distance $d$. Then, the write access cost is at least
\[
t\left( d + (t - 1)k + \left\lceil \frac{(t - 1)k}{r} \right\rceil - 2 \right) - (t - 1)n^F. 
\]

Write $\Delta = n_F - 2d - \left( (t - 1)k + \left\lceil \frac{(t - 1)k}{r} \right\rceil \right) + 3$. Then, the read access cost  is at least
\[
\left\{
\begin{array}{ll}
t k, & \text{if } \Delta \leq 0; \\
t \left(k - \left\lceil \frac{r\Delta}{r+1} \right\rceil \right), & \text{otherwise}.
\end{array}
\right. 
\]

Moreover, when $d > n_I - k + 1$, the read access cost is at least $t k$.
\end{theorem}

\begin{remark} \label{perred}
    For the read access cost, we can evaluate the efficiency of the conversion process in another way. For each $w \in \mathcal{W}_j$, the symbol indexed by $w$ in $\mathcal{C}^{F_j}$ can be expressed as a function of certain components in $\mathcal{C}^{I_1}|_{\mathcal{R}_{1,j}} \times \cdots \times \mathcal{C}^{I_{t_1}}|_{\mathcal{R}_{t_{1},j}}$, we denote these components as $\mathcal{C}^{I_1}|_{\mathcal{R}_{1,j,w}} \times \cdots \times \mathcal{C}^{I_{t_1}}|_{\mathcal{R}_{t_{1},j,w}}$, where $\mathcal{R}_{1,j,w}$ is a subset of $\mathcal{R}_{1,j}$. The per-symbol read access cost is then defined as $\rho'_r=\sum_{i \in [t_1] } \sum_{j \in [t_2]} \sum_{w \in \mathcal{W}_j} |\mathcal{R}_{i,j,w} |$. In practical storage applications, per-symbol read access cost also receives significant attention. A convertible code is said to be per-symbol read access-optimal if it minimizes the total per-symbol read access cost among all convertible codes with the same parameters. Trivially, the lower bound for per-symbol read access cost is always no less than the lower bound of read access cost.
\end{remark}

Since each initial code $\mathcal{C}^{I_i}$ has dimension $k_{I_i}$, reading $k_{I_i}$ information symbols suffices to recover all information in $\mathcal{C}^{I_i}$. Thus, we can always recover the entire message of $\mathcal{C}^{I_i}$ by reading only $k_{I_i}$ information symbols. In this case, the read access cost $\left|\bigcup_{j \in [t_2]} \mathcal{R}_{i,j}\right| = k_{I_i}$, and we refer to this approach of conversion as the {\bf default approach}.


\subsection{Rational Function Fields and Their Automorphism Groups}

We first introduce some fundamental facts about the rational function field. For more details, the reader may refer to \cite{stichtenoth2009algebraic}.

Let \( F \) be the rational function field \( \mathbb{F}_q(x) \), where \( x \) is transcendental over \( \mathbb{F}_q \). The field \( F \) has exactly \( q+1 \) rational places: the finite places \( P_{x-\alpha} \) corresponding to \( x-\alpha \) for each \( \alpha \in \mathbb{F}_q \), and the infinite place \( P_{\infty} \) with a uniformizer \( 1/x \). The set of all places of \( F \) is denoted by \( \mathbb{P}_F \).

For each place \( P \in \mathbb{P}_F \), let \( v_P \) be the normalized discrete valuation associated with \( P \). Given a nonzero function \( z \in F \), its zero divisor is defined as 
\[
(z)_0 = \sum_{P \in \mathbb{P}_F, v_P(z) > 0} v_P(z) P,
\]
and its pole divisor is given by  
\[
(z)_\infty = -\sum_{P \in \mathbb{P}_F, v_P(z) < 0} v_P(z) P.
\]
The principal divisor of \( z \) is then written as
\begin{equation} \label{prindiv} 
(z) = \sum_{P \in \mathbb{P}_F} v_P(z) P.
\end{equation}

We denote by \( \operatorname{Aut}(F/\mathbb{F}_q) \) the automorphism group of \( F \) over \( \mathbb{F}_q \), that is,
\[
\operatorname{Aut}(F/\mathbb{F}_q) = \{ \sigma: F \to F \mid \sigma \text{ is an } \mathbb{F}_q \text{-automorphism of } F \}.
\]

It is well known that any automorphism \( \sigma \in \operatorname{Aut}(F/\mathbb{F}_q) \) is uniquely determined by its action on \( x \), which takes the form  
\begin{equation}\label{action}
\sigma(x) = \frac{ax + b}{cx + d},
\end{equation}
where \( a, b, c, d \in \mathbb{F}_q \) and satisfy \( ad - bc \neq 0 \).

Let \( \text{GL}_2(q) \) denote the general linear group of \( 2 \times 2 \) invertible matrices over \( \mathbb{F}_q \). Each matrix  
\[
\begin{bmatrix} a & b \\ c & d \end{bmatrix} \in \text{GL}_2(q)
\]
induces an automorphism of \( F \) which acts on $x$ by the above action \eqref{action}. Two matrices in \( \text{GL}_2(q) \) define the same automorphism of \( F \) if and only if they lie in the same coset of \( Z(\text{GL}_2(q)) \), which consists of all scalar matrices of the form \( \{ a \text{I}_2 : a \in \mathbb{F}_q^* \} \). This leads to the isomorphism  
\[
\operatorname{Aut}(F/\mathbb{F}_q) \cong \text{PGL}_2(q) = \text{GL}_2(q)/Z(\text{GL}_2(q)).
\]
Thus, we can identify \( \operatorname{Aut}(F/\mathbb{F}_q) \) with the projective general linear group \( \text{PGL}_2(q) \). In our paper, we sometimes represent an automorphism in $\operatorname{Aut}(F/\mathbb{F}_q)$  by its corresponding matrix in \( \text{PGL}_2(q) \), via this isomorphism between the two groups.

Let \(\mathcal{G}\) be a subgroup of \(\operatorname{Aut}(F/\mathbb{F}_q)\). The fixed field of \(\mathcal{G}\) in $F$ is defined as
\[
F^{\mathcal{G}} = \{ z \in F : \sigma(z)=z \text{ for all } \sigma \in \mathcal{G} \}.
\]
By the Galois theory,  \(F/F^{\mathcal{G}}\) is a Galois extension with \(\operatorname{Gal}(F/F^{\mathcal{G}}) = \mathcal{G}\). Moreover, for any automorphism \(\sigma \in \operatorname{Gal}(F/F^{\mathcal{G}})\) and any place \(P \in \mathbb{P}_F\), the set $\sigma(P) = \{\sigma(z): z \in P \}$ is also a place of \(F\).

Let \(g(F^{\mathcal{G}})\) denote the genus of the field \(F^{\mathcal{G}}\). Then, the Hurwitz genus formula \cite[Theorem 3.4.13]{stichtenoth2009algebraic} states that
\begin{equation} \label{hurw}
2g(F)-2 = [F:F^{\mathcal{G}}](2g(F^{\mathcal{G}})-2) + \deg \operatorname{Diff}(F/F^{\mathcal{G}}),
\end{equation}
where \(\operatorname{Diff}(F/F^{\mathcal{G}})\) denotes the different of the extension \(F/F^{\mathcal{G}}\).

We will construct convertible codes using subgroups of $\text{PGL}_2(q)$. Therefore, it is essential to understand the structure of these subgroups. Fortunately, this structure is well established in the literature (see \cite{leemans2009polytopes} for details). In this paper, we first describe a general method to construct convertible codes based on arbitrary subgroups and then focus on specific subgroups to compute concrete examples. Thus, we list only a few particular subgroups of $\text{PGL}_2(q)$ here.

\begin{lemma}[\cite{de2010rank}] \label{subgps}
    Let $q=p^s$ for some prime $p$, the following groups are isomorphic to certain subgroups of $\mathrm{PGL}_2(q)$.
    \begin{itemize}
        \item  A cyclic group of order $d$, where $d$ is a divisor of $q+1$.
        \item  A semidirect product of an elementary abelian group of order $p^v$ and a cyclic group of order $u$ for every pair $(u,v)$ satisfying  $v\leq s$ and $u \mid \text{gcd~}(q-1,p^v-1)$.
        \item  A dihedral group of order $2d$, where $d$ is a divisor of either $q+1$ or $q-1$.
    \end{itemize}
\end{lemma}

\subsection{Algebraic Geometry Codes}

We now introduce algebraic geometry codes. For further background, the reader is referred to \cite{niederreiter2001rational}. Let \( F/\mathbb{F}_q \) be a function field of genus \( g \) with full constant field \(\mathbb{F}_q\). Suppose that
$\mathcal{D} = \{ P_1, P_2, \ldots, P_n \} \subseteq \mathbb{P}_F$
is a set of \( n \) distinct rational places of \( F \). For a divisor \( G \) on \( F/\mathbb{F}_q \) satisfying \( 2g-2 < \deg(G) < n \) and whose support is disjoint from \( \mathcal{D} \) (i.e., \(\operatorname{supp}(G) \cap \mathcal{D} = \emptyset\)), the algebraic geometry code associated with \( \mathcal{D} \) and \( G \) is defined by
\begin{equation} \label{alcod1}
\mathcal{C}_{\mathcal{L}}(\mathcal{D},G) \triangleq \{ (x(P_1),\ldots,x(P_n)) : x \in \mathcal{L}(G) \} \subseteq \mathbb{F}_q^n,
\end{equation}
where \(\mathcal{L}(G) =\{ z\in F^{*}: (z) \geq -G \} \cup \{ 0\}   \) is the Riemann–Roch space associated to $G$, which, by the Riemann's Theorem\cite[Theorem 1.4.17]{stichtenoth2009algebraic}, has dimension \(\deg(G)+1-g\). The code \(\mathcal{C}_{\mathcal{L}}(\mathcal{D},G)\) is an \([n, k, d]_q\) linear code with dimension \( k = \dim_{\mathbb{F}_q} \mathcal{L}(G) \) and minimum distance satisfying \( d \ge n - \deg(G) \) 
 (\cite[Theorem 2.2.2]{stichtenoth2009algebraic}). Furthermore, if \( V \) is any subspace of \(\mathcal{L}(G)\), we can define a subcode
\begin{equation} \label{alcod2}
\mathcal{C}_{\mathcal{L}}(\mathcal{D},V) \triangleq \{ (x(P_1),\ldots,x(P_n)) : x \in V \} \subseteq \mathbb{F}_q^n.
\end{equation}
In this case, the dimension of \(\mathcal{C}_{\mathcal{L}}(\mathcal{D},V)\) is equal to \(\dim_{\mathbb{F}_q} V\) and its minimum distance is still at least \( n - \deg(G) \).

When \(\operatorname{supp}(G) \cap \mathcal{D} \neq \emptyset\), the definition can be modified. Let \( m_i = v_{P_i}(G) \) for each \( P_i \in \mathcal{D} \), and choose a uniformizer \(\pi_{P_i}\) at \( P_i \). The modified algebraic geometry code is then defined as
\begin{equation} \label{alcod3}
\mathcal{C}_{\mathcal{L}}(\mathcal{D},G) \triangleq \{ ((\pi_{P_1}^{m_1} x)(P_1), \ldots, (\pi_{P_n}^{m_n} x)(P_n)) : x \in \mathcal{L}(G) \} \subseteq \mathbb{F}_q^n.
\end{equation}
This construction also yields an \([n, \deg(G)+1-g, ~\ge n-\deg(G)]_q\) linear code. Similarly, for any subspace \( V \subseteq \mathcal{L}(G) \), the corresponding subcode is defined by
\begin{equation} \label{alcod4}
\mathcal{C}_{\mathcal{L}}(\mathcal{D},V) \triangleq \{ ((\pi_{P_1}^{m_1} x)(P_1), \ldots, (\pi_{P_n}^{m_n} x)(P_n)) : x \in V \} \subseteq \mathbb{F}_q^n.
\end{equation}
The dimension of \(\mathcal{C}_{\mathcal{L}}(\mathcal{D},V)\) is equal to \(\dim_{\mathbb{F}_q} V\), and its minimum distance remains at least \(n-\deg(G)\).

\section{Lower Bound on the Access Cost} \label{sec3}

In this section, we establish a new lower bound on the access cost of generalized merge-convertible codes when the final code is an $(r,\delta)$-LRC, generalizing the known results in \cite{ge2024mds} and \cite{kong2024locally}. Since there is only one final code in the merge case, for clarity of notation, we use the superscript $F$ to denote the final code, and use $\mathcal{U}_i,\mathcal{R}_i$ and $\mathcal{W}$ to represent the index set of unchanged symbols in $\mathcal{C}^{I_i}$, read symbols in $\mathcal{C}^{I_i}$ and written symbols in $\mathcal{C}^F$ respectively.

We first give a upper bound on the number of unchanged symbols.
\begin{theorem}\label{III.1}
    Let $\mathscr{C}$ be a $(t, 1)_q$ generalized merge-convertible code over the finite field $\mathbb{F}_q$. Assume that the final code $\mathcal{C}^F$ has $(r, \delta)$-locality and let $d_F$ denote its minimum distance. Then, for every integer $i$ with $i \in [t]$, the following bound holds: 
    \begin{equation}   \label{wriacc}
    |\mathcal{U}_i| \leq k_{I_i}+  n_F-k_F-d_F+1-\left( \left\lceil \frac{k_F-k_{I_i}}{r} \right\rceil -1\right)(\delta-1).
    \end{equation}

\end{theorem}

\begin{proof}
     We present the proof for $i=1$, the arguments for the other cases are analogous.
    Let $\phi: \mathcal{C}^{I_1} \times \cdots \times \mathcal{C}^{I_{t}} \twoheadrightarrow \mathcal{C}^{F} $ be the bijective map associated with the code conversion. Let 
    \[\mathcal{C}'=\phi(\{ \mathbf{0} \} \times \mathcal{C}^{I_{2}} \times \cdots \times \mathcal{C}^{I_{t}})\subseteq \mathcal{C}^{F}.\]
    Then $\dim(\mathcal{C}')=k_F-k_{I_1}$ and $d(\mathcal{C^{'}}) \geq d_F$. By the commutative diagram \eqref{diag2}, we have
    \[\mathcal{C}'|_{\mathcal{U}_1}=\mathbf{0}.\]
    Then we can deduce that $\mathcal{C}'|_{\mathcal{U}^c_1}$ is an $(n_F-|\mathcal{U}_1|, k_F-k_{I_1},  \geq d_F;(r,\delta))$-LRC where $\mathcal{U}_1^c=(\{1\}\times[n_{I_1}])\backslash \mathcal{U}_1$. By the Singleton-type bound \eqref{Sbound2}, we have
    \[d_F \leq (n_F-|\mathcal{U}_1|)-(k_F-k_{I_1})+1-(\lceil \frac{k_F-k_{I_1}}{r} \rceil -1)(\delta-1).\]
    Thus,
    \[
  |\mathcal{U}_1| \leq k_{I_1}+  n_F-k_F-d_F+1-( \lceil \frac{k_F-k_{I_1}}{r} \rceil -1)(\delta-1).
    \]
\end{proof}

Before giving the lower bound on read access cost, we need two auxiliary lemmas.
\begin{lemma}[\cite{roth2006introduction}, Problem 2.8]\label{n-d+1}
    Suppose $\mathcal{C}$ is $[n,k,d]$-linear code, then for any $\Gamma \subseteq [n]$ with $|\Gamma|\geq n-d+1$, we have
    \[\dim(\mathcal{C}|_{\Gamma})=k.\]
\end{lemma}

\begin{lemma}\label{AT}
    Suppose $\mathcal{C}$ is an $(r, \delta)$-$\mathrm{LRC}$ of length $n$. Let $\mathcal{S}\subseteq [n]$ and $\Delta$ be a positive integer with $\Delta \leq |\mathcal{S}|$. Then there exist subsets $\mathcal{A}$ and $\mathcal{T}$ of $[n]$ satisfying that
\[ \mathcal{A}  \subseteq \mathcal{S} \cap \mathcal{T},~ |\mathcal{A}|=(\delta-1)\left\lfloor \frac{\Delta}{r+\delta-1}\right\rfloor,~|\mathcal{S}\cap \mathcal{T}| \leq \Delta\]
and $\mathcal{C}|_{\mathcal{T}}$ can be generated by $\mathcal{C}|_{\mathcal{T}\backslash \mathcal{A}}$.
\end{lemma}
\begin{proof}
    Since $\mathcal{C}$ is an $(r, \delta)$-LRC, we can choose some recovery sets $\mathcal{S}'_1, \mathcal{S}'_2, \cdots, \mathcal{S}'_{\ell} \subseteq [n]$ such that $[n]=\mathcal{S}'_1 \cup \mathcal{S}'_2 \cup \cdots\cup \mathcal{S}'_{\ell}$ with each $|\mathcal{S}'_i| \leq r+\delta-1$. Denote $\mathcal{S}_i=\mathcal{S}'_i\cap \mathcal{S}$. For each $i$, construct subsets $\mathcal{Q}_i$ as follows. Firstly, let
    \[\mathcal{Q}_1=\begin{cases}
           \mathcal{S}_1 , \text{ if } |\mathcal{S}_1|<\delta-1; \\ 
          \text{ any subset of }\mathcal{S}_1 \text{ of size }\delta-1, \text{ if } |\mathcal{S}_1|\geq\delta-1.
     \end{cases}\]
     Suppose $\mathcal{Q}_1, \cdots, \mathcal{Q}_j$ have already been constructed. Then define
     \[\mathcal{Q}_{j+1}=\begin{cases}
           \mathcal{S}_{j+1}\backslash(\mathcal{S}_1\cup \cdots \cup \mathcal{S}_j), \text{ if } |\mathcal{S}_{j+1}\backslash(\mathcal{S}_1\cup \cdots \cup \mathcal{S}_j)|<\delta-1; \\ 
          \text{ any subset of }\mathcal{S}_{j+1}\backslash(\mathcal{S}_1\cup \cdots \cup \mathcal{S}_j) \text{ of size }\delta-1, \text{ if } |\mathcal{S}_{j+1}\backslash(\mathcal{S}_1\cup \cdots \cup \mathcal{S}_j)|\geq\delta-1.
     \end{cases}\]
     By this process, we obtain $\ell$ pairwise disjoint subsets $\mathcal{Q}_1, \cdots, \mathcal{Q}_{\ell}$ with each $\mathcal{Q}_j \subseteq \mathcal{S}_j$ and $|\mathcal{Q}_j| \leq \delta-1$.
     Let $\mathcal{B}=\mathcal{Q}_1\cup \cdots \cup \mathcal{Q}_{\ell}$. We now aim to show that $|\mathcal{B}| \geq \lceil\frac{(\delta-1)|\mathcal{S}|}{r+\delta-1}\rceil$. 
     
     Let $\mathcal{Q}_{i_1}, \cdots, \mathcal{Q}_{i_a}$ be all subsets among $\mathcal{Q}_1, \cdots, \mathcal{Q}_{\ell}$ whose sizes are strictly smaller than $\delta-1$, where $1\leq i_1<\cdots<i_a \leq \ell$. Then $|\mathcal{Q}_j|=\delta-1$ for any $j \in [\ell]\backslash\{i_1,\cdots,i_a\}$. By the construction of $\mathcal{Q}_j$, we have 
     \[\mathcal{S}=\mathcal{Q}_{i_1}\cup \cdots\cup \mathcal{Q}_{i_a}\cup \big(\bigcup\limits_{j \in [\ell]\backslash\{i_1,\cdots,i_a\}} \mathcal{S}_{j}\big).\]
     By counting the sizes of $\mathcal{B}$ and $\mathcal{Q}$, we find that $|\mathcal{B}|=|\mathcal{Q}_{i_1}|+ \cdots+ |\mathcal{Q}_{i_a}|+(\ell-a)(\delta-1)$ and $|\mathcal{S}|\leq|\mathcal{Q}_{i_1}|+ \cdots+ |\mathcal{Q}_{i_a}|+(\ell-a)(r+\delta-1)$.
     Consequently, if we denote by $x$ the sum $|\mathcal{Q}_{i_1}|+ \cdots+ |\mathcal{Q}_{i_a}|$, we have,
     \[|\mathcal{B}|=x+(\ell-a)(\delta-1)\geq x+(\frac{|\mathcal{S}|-x}{r+\delta-1})(\delta-1)\geq \frac{(\delta-1)|\mathcal{S}|}{r+\delta-1}.\]
     Therefore, $|\mathcal{B}| \geq \lceil \frac{(\delta-1)|\mathcal{S}|}{r+\delta-1}\rceil$. 

Now, suppose $ \ell-a \geq \lfloor \frac{\Delta}{r+\delta-1}\rfloor$. Then we can choose a subset $\mathcal{J} \subseteq [\ell]\backslash\{i_1,\cdots,i_a\}$ with $|\mathcal{J}|=\lfloor \frac{\Delta}{r+\delta-1}\rfloor$ and define $\mathcal{A} =\bigcup_{j \in \mathcal{J}}\mathcal{Q}_j$. Clearly, $|\mathcal{A}|=(\delta-1)\lfloor \frac{\Delta}{r+\delta-1}\rfloor$. Let $\mathcal{T}=\bigcup_{j \in \mathcal{J}}\mathcal{S}'_j$. Then $|\mathcal{T}| \leq (r+\delta-1)\lfloor \frac{\Delta}{r+\delta-1}\rfloor \leq \Delta$. Note that each punctured code $\mathcal{C}|_{\mathcal{S}'_j}$ can be generated by $\mathcal{C}|_{\mathcal{S}'_j\backslash \mathcal{Q}_j}$, so by the construction of $\mathcal{Q}_j$, the code $\mathcal{C}|_{\mathcal{T}}$ can be generated by $\mathcal{C}|_{\mathcal{T}\backslash \mathcal{A}}$. 

On the other hand, if $\ell-a < \lfloor \frac{\Delta}{r+\delta-1}\rfloor$, then there exists an integer $a'$ with  $1 \leq a' \leq a$ and a subset $\mathcal{Q} \subseteq \mathcal{Q}_{i_{a'}}$, such that $|\mathcal{A}| =(\delta-1)\lfloor \frac{\Delta}{r+\delta-1}\rfloor$, where  
\[\mathcal{A}=\mathcal{Q}_{i_1}\cup \cdots\cup \mathcal{Q}_{i_{a'-1}} \cup \mathcal{Q} \cup \big(\bigcup\limits_{j \in [\ell]\backslash\{i_1,\cdots,i_a\}} \mathcal{Q}_{j}\big).\]
Define
\[\mathcal{T}=\mathcal{S}'_{i_1}\cup \cdots\cup \mathcal{S}'_{i_{a'-1}} \cup (\mathcal{S}'_{i_{a'}}\backslash \mathcal{Q}_{i_{a'}})\cup \mathcal{Q}\cup \big(\bigcup\limits_{j \in [\ell]\backslash\{i_1,\cdots,i_a\}} \mathcal{S}'_{j}\big).\]
Again, by the same reasoning, $\mathcal{C}|_{\mathcal{T}}$ can be generated by $\mathcal{C}|_{\mathcal{T}\backslash \mathcal{A}}$. Now, observe that
\[\mathcal{T}\cap \mathcal{S}=\mathcal{S}_{i_1}\cup \cdots\cup \mathcal{S}_{i_{a'-1}} \cup (\mathcal{S}_{i_{a'}}\backslash \mathcal{Q}_{i_{a'}})\cup \mathcal{Q}\cup \big(\bigcup\limits_{j \in [\ell]\backslash\{i_1,\cdots,i_a\}} \mathcal{S}_{j}\big).\] By the construction of $\mathcal{Q}_i$, we actually have
\[\mathcal{T}\cap \mathcal{S}=\mathcal{Q}_{i_1}\cup \cdots\cup \mathcal{Q}_{i_{a'-1}} \cup \mathcal{Q}\cup \big(\bigcup\limits_{j \in [\ell]\backslash\{i_1,\cdots,i_a\}} \mathcal{S}_{j}\big).\]
Denote $y=|\mathcal{Q}_{i_1}|+ \cdots+ |\mathcal{Q}_{i_{a'-1}}|+|\mathcal{Q}|$, then
\[(\delta-1)\left\lfloor \frac{\Delta}{r+\delta-1}\right\rfloor=|\mathcal{A}|=y+(\ell-a)(\delta-1),\]
and
\[|\mathcal{T}\cap \mathcal{S}|\leq y+(\ell-a)(r+\delta-1)=y+\frac{|\mathcal{A}|-y}{\delta-1}(r+\delta-1)\leq (r+\delta-1)\left\lfloor \frac{\Delta}{r+\delta-1}\right\rfloor\leq \Delta.\]
This completes the proof.
\end{proof}

Now, we can give a lower bound on the read access cost.

\begin{theorem}\label{III.2}
     Let $\mathscr{C}$ be a $(t, 1)_q$ generalized convertible code over the finite field $\mathbb{F}_q$. Assume that the final code $\mathcal{C}^F$ has $(r, \delta)$-locality and let $d_F$ denote its minimum distance. Define $\Delta_i = |\mathcal{U}_i \backslash \mathcal{R}_i| -d_F +1 $ for $i \in [t]$. Then, we have
     \begin{equation} \label{reaacc}
     |\mathcal{R}_i| \geq \begin{cases}
           k_{I_i} , \mathrm{ ~if ~} \Delta_i \leq 0; \\ k_{I_i} - \Delta_i +  (\delta-1) \lfloor \frac{\Delta_i}{r+\delta-1}  \rfloor, \mathrm{ ~otherwise}.
     \end{cases}
     \end{equation}
     for $i \in [t]$.
     Moreover, if $d_F > n_{I_i} -k_{I_i}+1$ for some $i \in [t]$, then $\Delta_i \leq 0$ and therefore $|\mathcal{R}_i| \geq k_{I_i}$.    
    
\end{theorem}

\begin{proof}
 We present the proof for $i=1$; the arguments for the other cases are analogous. Let $\mathcal{E}_i \subseteq \{i\}\times [n_{I_i}]$ be an information set of $\mathcal{C}^{I_i}$ for any $i \in [2,t]$. Then $\mathcal{C}^{I_i}$ can be generated by $\mathcal{C}^{I_i}|_{\mathcal{E}_i}$. By the two diagrams \eqref{diag1} and \eqref{diag2}, we have two surjective maps:
 \[\mathcal{C}^{I_1}|_{\mathcal{R}_{1}} \times \mathcal{C}^{I_{2}}|_{\mathcal{E}_2}  \times \cdots \times \mathcal{C}^{I_{t}}|_{\mathcal{E}_t} \twoheadrightarrow \mathcal{C}^{I_1}|_{\mathcal{R}_{1}} \times \mathcal{C}^{I_2}|_{\mathcal{R}_{2}} \times \cdots \times \mathcal{C}^{I_{t}}|_{\mathcal{R}_{t} } \twoheadrightarrow \mathcal{C}^{F}|_{\mathcal{W}},\]
 and
  \[\mathcal{C}^{I_1}|_{\mathcal{R}_{1}} \times \mathcal{C}^{I_{2}}|_{\mathcal{E}_2}\times \cdots \times \mathcal{C}^{I_{t}}|_{\mathcal{E}_t} \twoheadrightarrow \mathcal{C}^{I_1}|_{\mathcal{U}_1\cap \mathcal{R}_{1}} \times \mathcal{C}^{I_2}|_{\mathcal{U}_{2}} \times\cdots \times \mathcal{C}^{I_{t}}|_{\mathcal{U}_{t} } \stackrel{\mathrm{Id}}{\longrightarrow} \mathcal{C}^{F}|_{(\mathcal{U}_1\cap \mathcal{R}_{1}) \cup \mathcal{U}_{2} \cup \cdots \cup \mathcal{U}_{t}}.\]
  Thus,
  \[|\mathcal{R}_1|+k_{I_2}+\cdots+k_{I_t} \geq \dim(\mathcal{C}^{F}|_{(\mathcal{U}_1\cap \mathcal{R}_{1}) \cup \mathcal{U}_{2} \cup \cdots \cup \mathcal{U}_{t} \cup \mathcal{W}}).\]
If $\Delta_1 \leq 0$, i.e., $|\mathcal{U}_1 \backslash \mathcal{R}_1| \leq d_F-1$, then $|(\mathcal{U}_1\cap \mathcal{R}_{1}) \cup \mathcal{U}_{2} \cup \cdots \cup \mathcal{U}_{t} \cup \mathcal{W}|=n_F-|\mathcal{U}_1 \backslash \mathcal{R}_1|\geq n-d_F+1$. By Lemma \ref{n-d+1}, we have 
\[\dim(\mathcal{C}^{F}|_{(\mathcal{U}_1\cap \mathcal{R}_{1}) \cup \mathcal{U}_{2} \cup \cdots \cup \mathcal{U}_{t} \cup \mathcal{W}})=k_F=k_{I_1}+k_{I_2}+\cdots+k_{I_t}.\]
Then combining with the above inequality, we obtain that $|\mathcal{R}_1| \geq k_{I_1}$.

If $|\mathcal{U}_1 \backslash \mathcal{R}_1| > d_F-1$. By Lemma \ref{AT}, there exists a subset $\mathcal{A} \subseteq \mathcal{U}_1 \backslash \mathcal{R}_1$ of size $(\delta-1)\lfloor \frac{\Delta_1}{r+\delta-1}\rfloor$ and a subset $\mathcal{T}$ with $\mathcal{A}\subseteq \mathcal{T} \subseteq [n_{F}]$ such that
\[| (\mathcal{U}_1 \backslash \mathcal{R}_1)\cap \mathcal{T}| \leq \Delta_1\]
and $\mathcal{C}^{F}|_{\mathcal{T}}$ can be generated by $\mathcal{C}^{F}|_{\mathcal{T}\backslash \mathcal{A}}$. Then we have $|(\mathcal{U}_1 \backslash \mathcal{R}_1) \backslash \mathcal{T}| \geq d_F-1$. Let $\mathcal{D} \subseteq (\mathcal{U}_1 \backslash \mathcal{R}_1) \backslash \mathcal{T}$ with $|\mathcal{D}|=d_F-1$.  Now we have two surjective maps:
\[\mathcal{C}^{I_1}|_{\mathcal{R}_{1}} \times \mathcal{C}^{I_{2}}|_{\mathcal{E}_2}\times \cdots \times \mathcal{C}^{I_{t}}|_{\mathcal{E}_t} \twoheadrightarrow \mathcal{C}^{I_1}|_{\mathcal{R}_{1}} \times \mathcal{C}^{I_2}|_{\mathcal{R}_{2}} \times\cdots \times \mathcal{C}^{I_{t}}|_{\mathcal{R}_{t} } \twoheadrightarrow \mathcal{C}^{F}|_{\mathcal{W}},\]
 and
  \[\mathcal{C}^{I_1}|_{\mathcal{U}_1 \backslash (\mathcal{D}\cup \mathcal{A})} \times \mathcal{C}^{I_{2}}|_{\mathcal{E}_2}\times \cdots \times \mathcal{C}^{I_{t}}|_{\mathcal{E}_t} \twoheadrightarrow \mathcal{C}^{I_1}|_{\mathcal{U}_1 \backslash (\mathcal{D}\cup \mathcal{A})} \times \mathcal{C}^{I_2}|_{\mathcal{U}_{2}}\times \cdots \times \mathcal{C}^{I_{t}}|_{\mathcal{U}_{t} }  \stackrel{\mathrm{Id}}{\longrightarrow} \mathcal{C}^{F}|_{(\mathcal{U}_1 \backslash (\mathcal{D} \cup \mathcal{A})) \cup \mathcal{U}_{2} \cup \cdots \cup \mathcal{U}_{t}}.\]
  Combining them together, we obtain a surjective map:
\[\mathcal{C}^{I_1}|_{\mathcal{R}_1\cup (\mathcal{U}_1 \backslash (\mathcal{D} \cup \mathcal{A}))} \times \mathcal{C}^{I_{2}}|_{\mathcal{E}_2}\times \cdots \times \mathcal{C}^{I_{t}}|_{\mathcal{E}_t} \twoheadrightarrow \mathcal{C}^{F}|_{(\mathcal{U}_1 \backslash (\mathcal{D} \cup \mathcal{A})) \cup \mathcal{U}_{2} \cup \cdots \cup \mathcal{U}_{t}\cup \mathcal{W}}.\]
Since $\mathcal{T}\cap \mathcal{D}=\emptyset$, hence $\mathcal{T}\backslash \mathcal{A} \subset (\mathcal{U}_1 \backslash (\mathcal{D} \cup \mathcal{A})) \cup \mathcal{U}_{2} \cup \cdots \cup \mathcal{U}_{t}\cup \mathcal{W}$. Combining with the fact that $\mathcal{C}^{F}|_{\mathcal{T}}$ can be generated by $\mathcal{C}^{F}|_{\mathcal{T}\backslash \mathcal{A}}$, we have a surjective map:
\[\mathcal{C}^{F}|_{(\mathcal{U}_1 \backslash (\mathcal{D} \cup \mathcal{A})) \cup \mathcal{U}_{2} \cup \cdots \cup \mathcal{U}_{t}\cup \mathcal{W}} \twoheadrightarrow \mathcal{C}^{F}|_{\mathcal{T}\backslash \mathcal{A}}\twoheadrightarrow \mathcal{C}^{F}|_{\mathcal{T}}.\]
Therefore, we obtain
\[\mathcal{C}^{I_1}|_{\mathcal{R}_1\cup (\mathcal{U}_1 \backslash (\mathcal{D} \cup \mathcal{A}))} \times \mathcal{C}^{I_{2}}|_{\mathcal{E}_2}\times \cdots \times \mathcal{C}^{I_{t}}|_{\mathcal{E}_t} \twoheadrightarrow \mathcal{C}^{F}|_{(\mathcal{U}_1 \backslash (\mathcal{D} \cup \mathcal{A})) \cup \mathcal{U}_{2} \cup \cdots \cup \mathcal{U}_{t}\cup \mathcal{W}\cup \mathcal{T}}=\mathcal{C}^{F}|_{(\mathcal{U}_1 \backslash \mathcal{D}) \cup \mathcal{U}_{2} \cup \cdots \cup \mathcal{U}_{t}\cup \mathcal{W}}.\]

Note that $|(\mathcal{U}_1 \backslash \mathcal{D}) \cup \mathcal{U}_{2} \cup \cdots \cup \mathcal{U}_{t}\cup \mathcal{W}|=n_F-|\mathcal{D}|=n_F-d_F+1$, thus by Lemma \ref{n-d+1}, we have $\dim(\mathcal{C}^{F}|_{(\mathcal{U}_1 \backslash \mathcal{D}) \cup \mathcal{U}_{2} \cup \cdots \cup \mathcal{U}_{t}\cup \mathcal{W}})=k_F$. 
Therefore 
\[|\mathcal{R}_1\cup (\mathcal{U}_1 \backslash (\mathcal{D} \cup \mathcal{A}))|+k_{I_2}+\cdots+k_{I_t}\geq k_F\]
\[\Longrightarrow |\mathcal{R}_1|+|(\mathcal{U}_1  \backslash\mathcal{R}_1)\backslash (\mathcal{D}\cup \mathcal{A})|\geq k_{I_1}\]
\[\Longrightarrow |\mathcal{R}_1|\geq k_{I_1}-|(\mathcal{U}_1  \backslash\mathcal{R}_1)|+|(\mathcal{D}\cup \mathcal{A})|= k_{I_1}-(\Delta_1+d_F-1)+d_F-1+(\delta-1)\lfloor \frac{\Delta_1}{r+\delta-1}  \rfloor=k_{I_1}- \Delta_1 +  (\delta-1) \lfloor \frac{\Delta_1}{r+\delta-1}  \rfloor.\]

Moreover, if $d_F > n_{I_1} -k_{I_1}+1$, we suppose that $\Delta_1>0$, i.e., $|\mathcal{U}_1 \backslash \mathcal{R}_1| > d_F-1$. Let $\mathcal{D} \subseteq \mathcal{U}_1 \backslash \mathcal{R}_1$ with $|\mathcal{D}|=d_F-1$. Then we have a surjective map:
\[\mathcal{C}^{I_1}|_{\mathcal{D}^c} \times \mathcal{C}^{I_{2}}|_{\mathcal{E}_2}\times \cdots \times \mathcal{C}^{I_{t}}|_{\mathcal{E}_t} \twoheadrightarrow \mathcal{C}^{I_1}|_{\mathcal{R}_1\cup(\mathcal{U}_1 \backslash \mathcal{D})} \times \mathcal{C}^{I_{2}}|_{\mathcal{E}_2}\times \cdots \times \mathcal{C}^{I_{t}}|_{\mathcal{E}_t} \twoheadrightarrow \mathcal{C}^{F}|_{(\mathcal{U}_1 \backslash \mathcal{D}) \cup \mathcal{U}_{2} \cup \cdots \cup \mathcal{U}_{t}\cup \mathcal{W}},\]
where $\mathcal{D}^c=(\{1\}\times[n_{I_1}])\backslash \mathcal{D}$.

Note that $|(\mathcal{U}_1 \backslash \mathcal{D}) \cup \mathcal{U}_{2} \cup \cdots \cup \mathcal{U}_{t}\cup \mathcal{W}|=n_F-d_F+1$, then $\dim(\mathcal{C}^{F}|_{(\mathcal{U}_1 \backslash \mathcal{D}) \cup \mathcal{U}_{2} \cup \cdots \cup \mathcal{U}_{t}\cup \mathcal{W}})=k_F$ by Lemma \ref{n-d+1}. Therefore, 
\[|\mathcal{D}^c|+k_{I_2}+\cdots+k_{I_t}\geq k_F \Longrightarrow n_{I_1}-d_F+1\geq k_{I_1},\]  
which contradicts to the assumption.
\end{proof}

Combining the above two theorems, we can obtain a lower bound for access-optimal convertible code when the final code is an $(r, \delta)$-LRC.

\begin{theorem} \label{bouthm}
     Let $\mathscr{C}$ be a $(t, 1)_q$ generalized convertible code over the finite field $\mathbb{F}_q$. Assume that the final code $\mathcal{C}^F$ has $(r, \delta)$-locality and let $d_F$ denote its minimum distance. For each $i \in[t]$, denote $ \widetilde{\Delta_i} = k_{I_i}+  n_F-k_F-2d_F+2-( \lceil \frac{k_F-k_{I_i}}{r} \rceil -1)(\delta-1) $. Then the total access cost satisfies
     \[ \rho \geq \text{min }\rho_w + \text{min }\rho_r, \]
     where
     \begin{equation} \label{minwri1}
         \text{min }\rho_w = -(t-1)n_F+(t-1)k_F+td_F-t+ \sum_{i=1}^{t}( \lceil \frac{k_F-k_{I_i}}{r} \rceil -1)(\delta-1)
     \end{equation}
     and
     \begin{equation} \label{minrea1}
          \text{min }\rho_r = k_F - \sum_{i \in [t], ~\widetilde{\Delta_i} >0~and ~d_F \leq n_{I_{i}}-k_{I_i}+1}[\widetilde{\Delta_i} - (\delta-1) \left\lfloor \frac{\widetilde{\Delta_i}}{r+\delta-1}  \right\rfloor].
     \end{equation}
     
\end{theorem}

\begin{proof}
    In \eqref{wriacc}, we give the upper bound of unchanged symbols for each initial code. Therefore \eqref{minwri1} follows by the definition of written symbols.  Notice that $\Delta_i \leq |\mathcal{U}_i|-d_F+1 \leq \widetilde{\Delta_i} $,  when either $\widetilde{\Delta_i} \leq 0  $ or $d_F > n_{I_{i}}-k_{I_i}+1 $, we must read at least $k_{I_i}$ symbols in $\mathcal{C}^{I_i}$ by \eqref{reaacc}, hence \eqref{minrea1} follows from \eqref{reaacc}.
\end{proof}

\begin{remark} \label{remark49}
    Near the completion of this manuscript, we became aware of a similar result obtained independently in the very recent paper \cite{ge2025locally}. They considered the conversion between $(r,\delta)$-LRCs and derived a slightly different bound compared to the one in Theorem \ref{bouthm}. Nevertheless, our approach to the proof is different from theirs, and the constructions they provided also differ from ours.
\end{remark}

With appropriate choices of code parameters, the above access cost bound becomes more elegant. We have:
\begin{corollary} \label{rdelcor}
    Let $\mathscr{C}$ be a $(t, 1)_q$ generalized merge convertible code over the finite field $\mathbb{F}_q$, where $\mathcal{C}^{I_{i}}$ has parameters $[n_{I_i},k_{I_i}=k_ir]$ for each $i \in [t]$ and $\mathcal{C}^F$ is an optimal $(n_F,k_F=\sum\limits_{i=1}^{t}k_ir;(r,\delta))$-$\mathrm{LRC}$. Then the total access cost satisfies
     \[ \rho \geq \text{min }\rho_w + \text{min }\rho_r, \]
     where
     \begin{equation} \label{minwri2}
         \text{min }\rho_w = n_F-(r+\delta-1)\sum\limits_{i=1}^{t}k_i
     \end{equation}
     and
     \begin{equation} \label{minrea2}
          \text{min }\rho_r = k_F - \sum_{i \in [t], ~-\frac{n_F}{r
          +\delta-1} + k_i+ \frac{k_F}{r} >0  ~and ~ d_F \leq n_{I_{i}}-k_{I_i}+1} [k_F+ rk_i-\frac{rn_F}{r+\delta-1}].
     \end{equation}
     More specifically, if $n_F$ can be written as $(r+\delta-1)(\sum\limits_{i=1}^{t}k_i+\ell)$, \eqref{minwri2}~and~ \eqref{minrea2} can be simplified to 
    \begin{equation} \label{minwri3}
         \text{min }\rho_w = (r+\delta-1)\ell
     \end{equation}
     and

     \begin{equation} \label{minrea3}
         \text{min }\rho_r = r \sum\limits_{i=1,~k_i \leq \ell \text{~or~} d_F > n_{I_{i}}-k_{I_i}+1 }^{t}k_{i} + r \sum\limits_{i=1,~k_i > \ell \text{~and~} d_F \leq n_{I_{i}}-k_{I_i}+1 }^{t}\ell,
     \end{equation} 
     respectively. 
\end{corollary}

\begin{proof}
    We have $d_F =n_F-k_F+1-(\sum\limits_{i=1}^{t}k_i-1)(\delta-1) $ since $\mathcal{C}^F$ is an optimal $(r,\delta)$-LRC. By direct computation, we have 
    \[
    \text{min }\rho_w = n_F-(r+\delta-1)\sum\limits_{i=1}^{t}k_i,
    \]
    \[
    \widetilde{\Delta_i} = -n_F+ (k_i+\sum\limits_{i^{'}=1}^{t}k_{i^{'}})(r+\delta-1)-(\delta-1) 
    \]
    and
    \[
    \lfloor \frac{\widetilde{\Delta_i}}{r+\delta-1}  \rfloor =-\frac{n_F}{r+\delta-1}+(k_i+\sum\limits_{i^{'}=1}^{t}k_{i'} )-1.
    \] Hence all the result follows by straightforward calculation.
\end{proof}

\begin{remark}\label{remIII.2}
    Theorem \ref{bouthm} covers the known bound in \cite{kong2024locally} by setting $\mathcal{C}^{I_1}=\mathcal{C}^{I_2}=\cdots=\mathcal{C}^{I_t}$ and $\delta=2$, it also generalizes the result in \cite{ge2024mds} by considering the final code to be an optimal  $(n_F,k_F;(r=k_F,\delta=2))$ LRC (i.e., an MDS code), without requiring any of the initial codes to be MDS. Therefore, our bound is applicable in a very broad range of scenarios.
\end{remark}

\section{ Construction of Access-Optimal Generalized Merge-Convertible Code Between MDS codes }\label{sec4}
In this section, we focus on the construction of access-optimal $(t,1)_q$ MDS generalized merge-convertible codes. To simplify notation, we first assume that all initial codes have identical parameters. In this case, we use the superscript $I$ to denote the initial code, and the superscript $F$ to represent the final code. In Example \ref{mdsexa}, we will show that the construction remains valid even when the initial codes differ in their parameters.

We begin by introducing a general framework for constructing access-optimal MDS merge-convertible codes and then provide several examples to illustrate the construction. We will extend this framework to design $(r,\delta)$ merge-convertible codes for  $\delta \geq 2$ in the next section and show that our convertible codes achieve the lower bound of access cost established in Theorem \ref{bouthm}, hence they are also access-optimal convertible codes.

By Theorem \ref{mdsbou}, the lower bound of access cost can be obtained by the default approach when $l_{I} < l_F$ or $k_I < l_F$, hence we assume that $l_F \leq \text{min~}\{l_I, k_I\}$ throughout this section.

The idea underlying our construction work as follows. Let $F$ be a rational function field, and let $\mathcal{G}$ be a subgroup of Aut$(F/\mathbb{F}_q)$ with order $l$. By the Galois theory, $F/F^{\mathcal{G}}$ is a Galois extension with Galois group Gal$(F/F^{\mathcal{G}})=\mathcal{G}$. Suppose that $P_0, P_1, P_2, \cdots, P_m$ are rational places of $F^{\mathcal{G}}$ that split completely in the field extension $F/F^{\mathcal{G}}$. Denote by $P_{i,1}, P_{i,2},\cdots, P_{i,l}$ the rational places of $F$ lying over $P_i$ for  $i \in [0,m]$. Since $\mathcal{G}$ acts transitively on the rational places $P_{i,1}, P_{i,2},\cdots, P_{i,l}$ for each $i$, we can reorder these places such that, without loss of generality, there exists a unique element $\sigma_j \in \mathcal{G}$ satisfying $\sigma_j(P_{i,1})=P_{i,j}$ for $i \in [0,m]$ and $j \in [l]$. For any function $f \in F$, by \cite[Equation (8.3)]{stichtenoth2009algebraic}, we have
\[\sigma_j(f)(P_{i,j})=\sigma_j(f)\sigma_j(P_{i,1})=f(P_{i,1}) \text{~if~} v_{P_{i,1}}(f) \geq 0. \]
 
 In essence, we can select one rational place from each orbit under the action of $\mathcal{G}$ and collect them together as the evaluation points for the information symbols of the initial code. By exploiting the aforementioned property, we extend the set of evaluation points and adjust the evaluation functions accordingly so that the values at the newly added points remain consistent with those at the original evaluation points of the initial code. This strategy maximizes the number of unchanged symbols. Furthermore, we choose the redundancy symbols of the final code to form a complete orbit under the action of $\mathcal{G}$, ensuring that these symbols can be determined solely by reading the corresponding redundancy symbols from the initial code. See Theorem \ref{mdsthm} below for further details. 

To construct convertible codes using this idea, we need to analyze the splitting behavior of rational places of $F^{\mathcal{G}}$ in extension $F/F^{\mathcal{G}}$. Given that the automorphism group Aut$(F/\mathbb{F}_q)$ is quite large and contains several well-structured subgroups, we begin by presenting a general framework for constructing access-optimal MDS generalized merge-convertible codes using an arbitrary subgroup of Aut$(F/\mathbb{F}_q)$, relying on an estimate of the number of ramified rational places in $F$. We then demonstrate how to construct convertible codes using specific subgroups explicitly.

We now present the main theorem of this section.
\begin{theorem} \label{mdsthm} 
    Assume that there exists a subgroup $\mathcal{G}$ of $\mathrm{PGL}_2(q)$ of order $l$. Let $m$ be a positive integer satisfying $m+1 \leq \lfloor \frac{q-2l+3}{l}  \rfloor $. For some positive integers $k, t, l^{'}$ satisfying $k\in [m]$, $ t \in [l]$, $k+l^{'} \leq q+1$ and  $l\leq min\{k, l^{'} \}$, define $n_I = k+l^{'}$, $k_I = k $, $n_F= kt + l $ and $k_F=kt$.
    Then there exists an access-optimal  $(n_I, k_I; n_F, k_F )$   $\mathrm{MDS}$ merge-convertible code with read access cost $t l$ and write access cost $l$. Moreover, this construction is per-symbol read access-optimal.
\end{theorem}

\begin{proof}
    Let $F$ denote a rational function field over $\mathbb{F}_q$ and consider the field extension $F/F^{\mathcal{G}}$, denote by $ \mathcal{R} = \{ R_1, R_2, \cdots, R_s \}$ the set of all rational places of $F$ which are ramified in $F/F^{\mathcal{G}}$. By the Hurwitz genus formula \eqref{hurw}, we have
    \[
    2 g(F)-2 \geq l[2g(F^{\mathcal{G}})-2] + \sum^{s}_{i=1} d_{R_i}(F/F^{\mathcal{G}}) \mathrm{deg}(R_i) \geq -2l+s,
    \]
    so that $s \leq 2l-2$.
    
    For any rational place $P \notin \mathcal{R} $ of $F$, let $Q$ be its restriction to $F^{\mathcal{G}}$, since every automorphism in Aut$(F/\mathbb{F}_q)$ fixes $\mathbb{F}_q$, the relative degree $f(P|Q) =[\mathcal{O}_P/P:\mathcal{O}_Q/Q] =1 $, and ramification index $e(P|Q) = 1$. By the fundamental equality \cite[Theorem 3.1.11]{stichtenoth2009algebraic} $e(P|Q)f(P|Q)g(P|Q)=1$, it follows that $Q$ splits completely into $l$ rational places in $F$. Since there are $q+1$ rational places in total over $\mathbb{F}_q$, there exist at least $q-2l+3$ rational places in $F$ lying above completely split rational places in $F^{\mathcal{G}}$. Given that $m+1 \leq \lfloor \frac{q-2l+3}{l}  \rfloor $, we can select a rational place  $P_{\infty}$ of $F$ such that $P_{\infty} \notin P_{i,j}$, for $i \in [0,m]$ and $j \in [l]$ and denote $P_{\infty} \cap F^{\mathcal{G}}$ by $Q_{\infty}$. 
    Since $F$ is a rational function field, there exists a function $x \in F$ such that $(x)_{\infty} = P_{\infty}$ and hence we have $F=\mathbb{F}_q(x)$. 
    
    Next, we choose $k+1$ rational places $\{ P_0, P_1, P_2, \cdots, P_{k} \}$ in $F^{\mathcal{G}}$ that split completely in $F/F^{\mathcal{G}}$. As mentioned before, denote by $P_{i,1}, P_{i,2},\cdots, P_{i,l}$ all rational places of $F$ lying over $P_i$ with particular order for each $i \in [0,k]$, we can then use $\sigma_j$ to denote the unique element in  $\mathcal{G}$ such that $\sigma_j(P_{i,1})=P_{i,j}$ for $i \in [0,k]$ and $j \in [l]$. Let $\mathcal{B}=\{P_{0,1},P_{0,2},\cdots,P_{0,l}\}$, $\mathcal{A}_j=\{P_{1,j},P_{2,j},\cdots,P_{k,j}\}$ for $j \in [t]$ and $\mathcal{A} =\mathcal{A}_1 \cup \mathcal{A}_2 \cup \cdots \mathcal{A}_{t} $. Additionally, select $l^{'}-l$ places outside of $\mathcal{A}_1 \cup \mathcal{B}$ and denote this set by $\mathcal{B}^{'}$, then define the initial code as
    \[
    \mathcal{C}^I=\{ \big(f(P)\big)_{P \in \mathcal{A}_1 \cup \mathcal{B} \cup \mathcal{B}^{'}} : f\in \mathbb{F}^{<k}_q[x] \}.
    \]
    Clearly,  $\mathcal{C}^I$ is an $[n_I, k_I]$ MDS code. 

     Next, for any $f_1,f_2,\cdots, f_{t} \in \mathbb{F}^{<k}_q[x]$, we define the transformation map as
    \[
    T(f_1,f_2,\cdots, f_{t})=   \sum^{t}_{j=1} g^{k-1}_j  h_{\mathcal{A} \backslash \mathcal{A}_j} \sigma_j(f_j),
    \]
    where $h_{\mathcal{A} \backslash \mathcal{A}_i}= \prod\limits_{a \in \mathcal{A} \backslash \mathcal{A}_i }(x-a)$ is the annihilator polynomial of the set $\mathcal{A} \backslash \mathcal{A}_i$, $g_j$ is a fixed function with divisor $(g_j) = \sigma_j(P_{\infty}) - P_\infty$, and such a function always exists in the rational function field. In particular, we choose $g_1=1$.
    Then we define the coefficient vector $ u = (u_P)_{P \in \mathcal{A}_1\cup \mathcal{A}_2 \cup  \cdots \cup  \mathcal{A}_{t} \cup \mathcal{B}}$ by:
     \[
     u_{P}=
     \begin{cases}
      [ \frac{1}{g^{k-1}_{i}} h^{-1}_{ \mathcal{A} \backslash \mathcal{A}_{i} } ](P) , & \text{if~} P \in \mathcal{A}_{i} \text{~for~} i \in [t]; \\
      1, & \text{if~} P \in \mathcal{B}.
     \end{cases}
     \]     
     Finally, define the final code  $\mathcal{C}^F$ as:
     \[
     \mathcal{C}^F = \{ \big(u_P T(f_1, f_2, \cdots, f_{t}) (P)\big)_{P \in  \mathcal{A} \cup \mathcal{B}} :f_1, f_2, \cdots, f_{t} \in \mathbb{F}^{  {<k}}_{q}[x] \}.
     \]
    
     Now we show that this definition is well-defined.
     We need to evaluate $T(f_1, f_2, \cdots, f_{t})$ at the places in $ \mathcal{A} \cup \mathcal{B}$. Note that $\mathbb{F}_q^{<k}[x] = \mathcal{L}((k-1)P_{\infty})$ and $\sigma(\mathcal{L}(G))=\mathcal{L}(\sigma(G))$ \cite[Lemma 1]{10268243} for any automorphism $\sigma \in \text{Aut}(F/\mathbb{F}_q)$ and any divisor $G$. Therefore $\sigma_j(f_j) \in \mathcal{L}((k-1)\sigma_j(P
     _{\infty})) $. After multiplication by $g^{k-1}_j$, the only possible pole of $g^{k-1}_j  \sigma_j(f_j)$ is $P_{\infty}$. Meanwhile, the only possible pole of $ h_{\mathcal{A} \backslash \mathcal{A}_j}$ is $P_{\infty}$. As a result, $T(f_1, f_2, \cdots, f_{t})$ is regular at all places in $ \mathcal{A} \cup \mathcal{B}$, which ensures that $\mathcal{C}^F$ is well-defined. 
     
To illustrate the conversion procedure more clearly, we distinguish the initial codes in the proof. Specifically, we assume that $\mathcal{C}^{I_i}=\mathcal{C}^I=\{ (f_i(P))_{P \in \mathcal{A}_1 \cup \mathcal{B} \cup \mathcal{B}^{'}}: f_i \in \mathbb{F}^{<k}_q[x] \}$ for $i \in [t]$. We now present the conversion map
 \[\phi: \mathcal{C}^{I_1}\times \mathcal{C}^{I_{2}}\times\cdots \times \mathcal{C}^{I_t} \longrightarrow \mathcal{C}^F\] which is defined by\[\phi\left(\big((f_1(P))_{P\in\mathcal{A}_1 \cup \mathcal{B} \cup \mathcal{B}^{'}}, \cdots,(f_t(P))_{P\in\mathcal{A}_1 \cup \mathcal{B} \cup \mathcal{B}^{'}}\big)\right)=\big(u_P T(f_1, f_2, \cdots, f_{t}) (P)\big)_{P \in  \mathcal{A} \cup \mathcal{B}} \]
 for any $f_1,f_2,\cdots, f_{t} \in \mathbb{F}^{<k}_q[x]$.
 
     Note that for any $P_{i,j} \in \mathcal{A}$, 
     \begin{align*} 
         u_{P_{i,j}}T(f_1,f_2,\cdots, f_{t}) (P_{i,j}) &=  [u_{P_{i,j}} \sum^{t}_{j^{'}=1} g^{k-1}_{j^{'}}  h_{\mathcal{A} \backslash \mathcal{A}_{j^{'}}} \sigma_{j^{'}}(f_{j^{'}})] (P_{i,j}) \\
          &= [u_{P_{i,j}} g^{k-1}_j  h_{\mathcal{A} \backslash \mathcal{A}_j} \sigma_j(f_j)] (P_{i,j}) \\
          &= [u_{P_{i,j}} g^{k-1}_j  h_{\mathcal{A} \backslash \mathcal{A}_j} \sigma_j(f_j)] (\sigma_j(P_{i,1})) \\
          &= f_j(P_{i,1}).
     \end{align*}
From the above equation, $\phi$ induces identity maps from  $\mathcal{C}^{I_j}|_{\mathcal{A}_1}$ to $\mathcal{C}^{F}|_{\mathcal{A}_j}$ for $j \in [t]$. Thus, the symbols in $\mathcal{C}^F$ corresponding to the places in $\mathcal{A}$ are the unchanged symbols, and hence the symbols corresponding to places in $\mathcal{B}$ are the written symbols. Thus, the write access cost is $|\mathcal{B}|=l$.  
Moreover, since the initial code is MDS, we can recover all polynomials $f_1,f_2,\cdots, f_{t}$ from  $ u_{P_{i,j}}T(f_1,f_2,\cdots, f_{t})(P_{i,j}) $, which implies that the final code has dimension $kt$. Since $T(f_1, f_2, \cdots, f_{t}) \in \mathcal{L}((k-1 )P_\infty+(kt-k)P_\infty)=\mathcal{L}((kt-1)P_\infty)$, it can have at most $kt-1$ zeros. Therefore, the minimum distance $d_F$ of the final code satisfies $d_F \geq n_F- (kt-1)=l+1 $. By the Singleton bound, this confirms that the final code is an $[n_F,k_F]$ MDS code.
     
     Next, we show that the per-symbol read cost of this construction is $t$.
     In fact, for any written symbol $P_{0,j} \in \mathcal{B}$, 
     \begin{align*}
         u_{P_{0,j}}T(f_1,f_2,\cdots, f_{t}) (P_{0,j}) &=  [ \sum^{t}_{j^{'}=1} g^{k-1}_{j^{'}}  h_{\mathcal{A} \backslash \mathcal{A}_{j^{'}}} \sigma_{j^{'}}(f_{j^{'}})] (P_{0,j}) \\
          &= [ \sum^{t}_{j^{'}=1} g^{k-1}_{j^{'}}  h_{\mathcal{A} \backslash \mathcal{A}_{j^{'}}}](P_{0,j}) [\sigma_{j^{'}}(f_{j^{'}})] (P_{0,j})   \\
          &=  [ \sum^{t}_{j^{'}=1} g^{k-1}_{j^{'}}  h_{\mathcal{A} \backslash \mathcal{A}_{j^{'}}}](P_{0,j}) f_{j^{'}} (\sigma^{-1}_{j^{'}} (P_{0,j})). \\
     \end{align*}
    Hence $u_{P_{0,j}}T(f_1,f_2,\cdots, f_{t}) (P_{0,j})$ is a linear combination of $f_1(P_{0,j})$, $f_2(\sigma^{-1}_2 (P_{0,j}))$,$\cdots$, $f_{t}(\sigma^{-1}_{t} (P_{0,j}))$ for $j \in [l] $.
    We can recover every written symbol by $t$ symbols of initial codes, and hence the total read symbols are equal to $t l$.
    
    In conclusion, by \cite[Theorem 1]{ge2024mds} or Remark \ref{remIII.2}, our construction is access-optimal, and which is also per-symbol read access-optimal defined in Remark \ref{perred}.     
    \end{proof}

\begin{remark} \label{remmds}
~
    \begin{enumerate} [label=\arabic*.]
        \item[(i)] To maximize the number of codewords that can be merged, we exclusively evaluate at completely split rational places. The places associated with the redundancy symbols in the final code are selected from a single complete orbit under the action of $\mathcal{G}$. This choice is motivated by the assumption $l_F \leq \text{min~}\{l_I, k_I\}$, which encourages minimizing the number of such symbols. Naturally, it is also possible to select multiple complete orbits as the places associated with the redundancy symbols; as long as the condition $l_F \leq \text{min~}\{l_I, k_I\}$ is satisfied, the proposed conversion process outperforms the default approach. In fact, redundancy symbols can also be associated with ramified rational places, which means that the number of redundancy symbols is not necessarily a multiple of the size of any subgroup of $\text{PGL}_2(q)$. This provides greater flexibility in parameter selection, without decreasing the maximum number of codewords that can be merged (since the information symbols are still associated with completely split rational places). We will explain this point more precisely in Remark \ref{remram} in the next chapter.
        \item[(ii)] As long as the rational places in at least one orbit are not used for evaluation, we can define $P_\infty$ and $Q_\infty$ in the proof of Theorem \ref{mdsthm} and construct convertible codes as above. However, by using modified algebraic geometry codes \eqref{alcod4}, all rational places can in fact be used for evaluation. We will illustrate this point in Corollary \ref{poleva}.
        \item[(iii)] Through simple modifications, the parameters of the initial codes do not need to be the same, hence we can construct an access-optimal MDS generalized merge-convertible code by our method. We will explain this in Example \ref{mdsexa}.
        \item[(iv)] If a subgroup $\mathcal{G}$ of the multiplicative group $\mathbb{F}_q^*$ is used in the construction, the method for selecting rational places corresponding to information symbols in the access-optimal MDS merge-convertible codes in \cite[Corollary II.2]{kong2024locally} is essentially the same as ours. However, our approach is per-symbol read access-optimal. Furthermore, our construction only requires the existence of a subgroup in the automorphism group of a rational function field, which allows for more flexible and broader achievable parameters compared to those in \cite{kong2024locally}.
    \end{enumerate}
\end{remark}

Here, we present an explicit construction of an access-optimal MDS merge-convertible code using a cyclic subgroup of order dividing $q+1$ in Aut$(F/\mathbb{F}_q)$. This construction explains how all rational places can be utilized to construct convertible codes. 

\begin{corollary} \label{poleva}
    
    Let $l, l^{'}, k$ be positive integers such that $l~|~q+1$, $k \leq \frac{q+1}{l}-1 $, $ t \in [l] $, $l \leq \text{min~}\{ k , l^{'} \} $ and $ k + l^{'} \leq q+1$. Define $n_I = k+l^{'}$, $k_I = k $, $n_F= kt + l $ and $k_F=kt$. Then there exists an access-optimal $(n_I, k_I; n_F,k_F )$  $\mathrm{MDS}$ merge-convertible code over $\mathbb{F}_q$.  

\end{corollary}

\begin{proof}
    
 We prove the result for the specific case $k = \frac{q+1}{l} - 1 \text{~and~}t = l$,
as this choice yields the maximum code length of $q+1$ over the finite field $\mathbb{F}_q$. Let $F = \mathbb{F}_q(x)$
be the rational function field, and let $f(x) = x^2 + ax + b \in \mathbb{F}_q[x]~$be a primitive polynomial of order $q^2 - 1$. Define the automorphism $\eta$ of $F$ by $\eta(x) = \frac{1}{-bx - a}$.
According to \cite[Lemma V.1]{8865641}, $\langle \eta \rangle$ is a cyclic subgroup of Aut$(F/\mathbb{F}_q)$ of order $q+1$, and there exists a unique place $P$ of degree 2 in $F$ that is totally ramified in the extension $F/F^{\langle \eta \rangle}$. Applying the Hurwitz genus formula \eqref{hurw}, we obtain
\[
2q = 2|\langle \eta \rangle| - 2 = \deg \operatorname{Diff}(F/F^{\langle \eta \rangle}).
\]
Since $P$ is totally ramified, we have
\[
d_P(F/F^{<\eta>}) \cdot \deg P = \bigl[e_P(F/F^{<\eta>}) - 1\bigr] \deg P = 2q = \deg \operatorname{Diff}(F/F^{\langle \eta \rangle}).
\]
This shows that $P$ is the only ramified place in the extension $F/F^{\langle \eta \rangle}$.

Consequently, every rational place of $F$ lies above a completely split rational place in the extension $F/F^{\langle \eta \rangle}$, and we conclude that all $q+1$ rational places of $F$ form a single cyclic orbit under the action of $\eta$. Consider the automorphism $\sigma \triangleq \eta^{\frac{q+1}{l}}$, it is clear that $\langle \sigma \rangle$ form a subgroup of Aut$(F/\mathbb{F}_q)$ with order $l$. We can rename all the rational places in $F$ as $P_{i,j}$ for $ i \in [0,k]$ and $ j \in [l]$ such that $\sigma^s\eta^t (P_{0,1})=P_{t,s+1}$ for $ t \in [0,k]$ and $ s \in [0,l-1]$. Formally, let $\sigma_j$ in the proof of Theorem \ref{mdsthm} be defined as $\sigma^{j-1}$ for $j \in [l]$, and choose $P_{0,1} = P_{\infty}$ as in the theorem. Then, we define both the initial and final codes as stated in the proof of Theorem \ref{mdsthm}. It remains to show that by suitably modifying the definitions, both the initial and final codes can be evaluated at $ P_{0,1} $ without affecting the MDS property, while also ensuring that each redundancy symbol in the final code can still be obtained by reading exactly one redundancy symbol from every initial code.

First, we only need to explain how the initial code is evaluated at ${P_{0,1}}$, here we adopt the modified algebraic geometry code described in \eqref{alcod4}, where we fix a uniformizer $\pi_{P_{0,1}}$ for ${P_{0,1}}$ and evaluate at ${P_{0,1}}$ by computing $\pi^{k-1}_{P_{0,1}}f(P_{0,1})$. Next, when evaluating at the point ${P_{0,1}}$ in the final code, we compute $\pi^{kt-1}_{P_{0,1}}T(f_1,f_2,\cdots,f_t)(P_{0,1})$. Since $T(f_1,f_2,\cdots,f_t) $ belongs to the space $\mathcal{L}((kt-1)P_{0,1})$, the definition is well-posed and consistent. Also, by considering how many zeros a codeword can have, as shown in the proof of Theorem \ref{mdsthm}, the modified codes are still MDS codes.

Next, we have 
    \begin{align*}
         \pi^{kt-1}_{P_{0,1}}u_{P_{0,1}}T(f_1,f_2,\cdots, f_{t}) (P_{0,1}) 
         &= \pi^{kt-1}_{P_{0,1}}[ h_{\mathcal{A} \backslash \mathcal{A}_1} \sigma_1(f_1)] (P_{0,1}) + [\pi^{kt-1}_{P_{0,1}} \sum^{t}_{j^{'}=2} g^{k-1}_{j^{'}}  h_{\mathcal{A} \backslash \mathcal{A}_{j^{'}}} \sigma_{j^{'}}(f_{j^{'}})] (P_{0,1}) \\
          &=[ \pi^{kt-1}_{P_{0,1}} h_{\mathcal{A} \backslash \mathcal{A}_1} f_1] (P_{0,1}) + \sum^{t}_{j^{'}=2}  [  \pi^{kt-k}_{P_{0,1}}  h_{\mathcal{A} \backslash \mathcal{A}_{j^{'}}}](P_{0,1}) [ \pi^{ k-1 }_{P_{0,1}} g^{k-1}_{j^{'}} \sigma_{j^{'}}(f_{j^{'}})] (P_{0,1})   \\
          &= [ \pi^{kt-1}_{P_{0,1}} h_{\mathcal{A} \backslash \mathcal{A}_1} f_1] (P_{0,1}) + \sum^{t}_{j^{'}=2}  [  \pi^{kt-k}_{P_{0,1}}  h_{\mathcal{A} \backslash \mathcal{A}_{j^{'}}}](P_{0,1}) [ \pi^{ k-1 }_{P_{0,1}} g^{k-1}_{j^{'}}](P_{0,1})  f_{j^{'}} ( \sigma^{-1}_{j^{'}} (P_{0,1})) . \\
     \end{align*}
     Since $v_{P_{0,1}}(\pi^{kt-1}_{P_{0,1}} h_{\mathcal{A}\backslash \mathcal{A}_1} f_1) \geq 0$, the function can be evaluated at $P_{0,1}$. Moreover, if $v_{P_{0,1}}(\pi^{kt-1}_{P_{0,1}} h_{\mathcal{A} \backslash \mathcal{A}_1} f_1) = 0$, it is clear that the expression $[ \pi^{kt-1}_{P_{0,1}} h_{\mathcal{A} \backslash \mathcal{A}_1} f_1] (P_{0,1})$ is a nonzero scalar multiple of $\pi^{ k-1}_{P_{0,1}} f_1 (P_{0,1})$ and this scalar is independent of the choice of $f_1$. Furthermore, since the product $[  \pi^{kt-k}_{P_{0,1}}  h_{\mathcal{A} \backslash \mathcal{A}_{j^{'}}}](P_{0,1}) [ \pi^{ k-1 }_{P_{0,1}} g^{k-1}_{j^{'}}](P_{0,1}) $ is also a nonzero scalar (the valuations $v_{P_{0,1}}$ of both functions are zero), it follows that $\pi^{kt-1}_{P_{0,1}}u_{P_{0,1}}T(f_1,f_2,\cdots, f_{t}) (P_{0,1})$ is indeed a linear combination of \\
       $\pi^{ k-1}_{P_{0,1}} f_1 (P_{0,1}), f_{2} ( \sigma^{-1}_{2} (P_{0,1})),\cdots, f_{t} ( \sigma^{-1}_{t} (P_{0,1}))  $.
     
    Also, for $j \in [2,t] $ we have
    \begin{align*}
         u_{P_{0,j}}T(f_1,f_2,\cdots, f_{t}) (P_{0,j}) &= [ \sum^{t}_{j^{'}=1} g^{k-1}_{j^{'}}  h_{\mathcal{A} \backslash \mathcal{A}_{j^{'}}} \sigma_{j^{'}}(f_{j^{'}})] (P_{0,j}) \\
          &= \sum^{t}_{j^{'}=1, j^{'}\neq j}[  g^{k-1}_{j^{'}}  h_{\mathcal{A} \backslash \mathcal{A}_{j^{'}}}](P_{0,j}) [\sigma_{j^{'}}(f_{j^{'}})] (P_{0,j})   
          +   [  g^{k-1}_{j}  h_{\mathcal{A} \backslash \mathcal{A}_j}](P_{0,j}) [\sigma_{j}(f_{j})] (P_{0,j})   \\
          &=   \sum^{t}_{j^{'}=1, j^{'}\neq j}[  g^{k-1}_{j^{'}}  h_{\mathcal{A} \backslash \mathcal{A}_{j^{'}}}](P_{0,j}) f_{j^{'}} (\sigma^{-1}_{j^{'}}(P_{0,j}))  
          +    h_{\mathcal{A} \backslash \mathcal{A}_j}(P_{0,j}) [ g^{k-1}_{j} \sigma_{j}(f_{j})] (P_{0,j}). 
     \end{align*}
    
    Similarly, it can be seen that $u_{P_{0,j}}T(f_1,f_2,\cdots, f_{t}) (P_{0,j})$ is a linear combination of  $f_{j^{'}}(\sigma^{-1}_{j^{'}}(P_{0,j}))$ for $j^{'} \in  \{1,2,\cdots,t\}\backslash \{j\}$ and $ [g^{k-1}_{j} \sigma_{j}(f_{j})] (P_{0,j})$. Since $v_{P_{0,j}}(g^{k-1}_{j} \sigma_{j}(f_{j})) \geq 0$, the function can be evaluated at $P_{0,j}$. Moreover, if $v_{P_{0,j}}(g^{k-1}_{j} \sigma_{j}(f_{j})) = 0$, the resulting expression is a nonzero scalar multiple of $\pi^{ k-1}_{P_{0,1}} f_j (P_{0,1})$ and this scalar is independent of the choice of $f_j$. We conclude  that $u_{P_{0,j}}T(f_1,f_2,\cdots, f_{t}) (P_{0,j})$ is a linear combination of  $f_{j^{'}}(\sigma^{-1}_{j^{'}}(P_{0,j}))$ for $j^{'} \in  \{1,2,\cdots,t\}\backslash \{j\}$ and $\pi^{ k-1}_{P_{0,1}} f_j (P_{0,1}) $.
    
    The rest of the proof is similar to Theorem \ref{mdsthm}.
 \end{proof}

We give an example here to illustrate our construction. By the way, we explain that the parameters of initial codes need not to be the same. 
\begin{example} \label{mdsexa}
    For $q=23$, we provide an explicit construction of an access-optimal MDS generalized merge-convertible code.
    First of all, we choose a quadratic primitive polynomial $x^2-2x+5=0$ in $\mathbb{F}_{23}[x]$, let $\eta$ be the automorphism of $\mathbb{F}_{23}$ which sends $x$ to $\frac{1}{-2x+5}$, then $ \langle \eta \rangle$ forms a group of order 24, which acts cyclically on rational places of $\mathbb{F}_{23}$. We choose subgroup $\mathcal{G}=\langle \eta^6 \rangle$ of $\langle \eta \rangle$, which is a cyclic group of order 4. Consider the field extension $\mathbb{F}_{23} (x)/\mathbb{F}_{23} (x)^{\mathcal{G}}$, 
    let
    \[z = x+\sigma^6(x)
    +\sigma^{12}(x)+\sigma^{18}(x) = x + \frac{3x-1}{5x+1}+\frac{4x-4}{-3x-4} + \frac{-8x+15}{-6x-1} = \frac{x^4+8x^2+4x+7}{x^3+4x^2+11x+21}. \]
    Clearly, $z$ is fixed by $\mathcal{G}$, thus we have $\mathbb{F}_{23}(z) \subset \mathbb{F}_{23} (x)^{\mathcal{G}} $. Also, since $[\mathbb{F}_{23}(x):\mathbb{F}_{23}(z)]= \text{~max~} \{ \text{deg~} (x^4+8x^2+4x+7), \text{~deg~}(x^3+4x^2+11x+21) \} = 4$, 
     we have $\mathbb{F}_{23} (x)^{\mathcal{G}}=\mathbb{F}_{23}(z)$. 
     
     Six rational places of $\mathbb{F}_q (z)$, that is, the infinite place ${\infty}$ of $z$ and the zero places of ${z-12, z-3, z-10, z-17, z -8}$, completely split into $\{ P_{\infty}, P_{9}, P_{14}, P_{19} \}$, $\{ P_{20}, P_{5}, P_{6}, P_{4} \}$, $\{ P_{2}, P_{16}, P_{18}, P_{13} \}$, $\{ P_{21}, P_{7}, P_{17}, P_{11} \}$, $\{ P_{12}, P_{3}, P_{15}, P_{10} \}$,\\ $\{ P_{0}, P_{8}, P_{1}, P_{22} \}$ in $\mathbb{F}_{23} (x)$ respectively. 
     
     Let $\mathcal{B} = \{ P_{\infty}, P_{9}, P_{14}, P_{19}\}$, $\mathcal{A}_1 = \{ P_{20}, P_{2}, P_{21}, P_{12}, P_{0} \}$, $\mathcal{A}_2 = \{ P_{5}, P_{16}, P_{7}, P_{3}, P_{8} \}$, \\
     $\mathcal{A}_3 = \{ P_{6}, P_{18}, P_{17}, P_{15}, P_{1} \}$, $\mathcal{A}_4 = \{ P_{4}, P_{13}, P_{11}, P_{10}, P_{22} \}$, we define the initial code  $\mathcal{C}^I$ as:
     \[
     \mathcal{C}^{I_i}=\{ \big(f(P)\big)_{P \in \mathcal{A}_1 \cup \mathcal{B}} : f\in \mathbb{F}^{ < 5}_{23}[x] \} \text{~for~} i \in [1,3],
     \]
     \[
     \mathcal{C}^{I_4}=\{ \big(f(P)\big)_{P \in (\mathcal{A}_1 \backslash P_{20})  \cup \mathcal{B}}  : f\in \mathbb{F}^{ < 4}_{23}[x] \},
     \]
     define the transformation map $T$ from $\mathbb{F}^{ < 5}_{23}[x] \times \mathbb{F}^{ < 5}_{23}[x] \times \mathbb{F}^{ < 5}_{23}[x] \times \mathbb{F}^{ < 4}_{23}[x] $ to $\mathbb{F}_{23} [x]$ as:
     \[(f_1, f_2, f_3 ,f_4) \longmapsto h_{ \mathcal{A} \backslash (A_1 \cup P_4) }f_1 +  h_{ \mathcal{A} \backslash (\mathcal{A}_2 \cup P_4) } (5x+1)^4 \eta^{6}(f_2) + h_{ \mathcal{A} \backslash (\mathcal{A}_3 \cup P_4) } (-3x-4)^4 \eta^{12}(f_3) + h_{ \mathcal{A} \backslash \mathcal{A}_4 } (-6x-1)^4 \eta^{18}(f_4) \]
     and define the coefficient vector $ u = (u_P)_{P \in \mathcal{A}_1\cup \mathcal{A}_2 \cup  \mathcal{A}_{3} \cup  \mathcal{A}_{4} \cup \mathcal{B}}  $ as follow:
     \[
     u_{P}=
     \begin{cases}
      [h^{-1}_{ \mathcal{A} \backslash (\mathcal{A}_1 \cup P_4) }](P), & \text{if~} P \in \mathcal{A}_1; \\
      [\frac{1}{(5x+1)^4} h^{-1}_{ \mathcal{A} \backslash (\mathcal{A}_2 \cup P_4) } ](P), & \text{if~} P \in \mathcal{A}_2; \\
      [ \frac{1}{(-3x-4)^4} h^{-1}_{ \mathcal{A} \backslash (\mathcal{A}_3 \cup P_4) } ](P), & \text{if~} P \in \mathcal{A}_3; \\
      [ \frac{1}{(-6x-1)^4} h^{-1}_{ \mathcal{A} \backslash \mathcal{A}_4 } ](P) , & \text{if~} P \in \mathcal{A}_4 \backslash P_4; \\
      1, & \text{if~} P \in \mathcal{B}.
     \end{cases}
     \]     
     then define the final code  $\mathcal{C}^F$ as:
     \[
     \mathcal{C}^F = \{ \big( u_P T(f_1, f_2, f_3, f_4) (P) \big)_{P \in  \mathcal{A}_1\cup \mathcal{A}_2 \cup  \mathcal{A}_{3} \cup  (\mathcal{A}_{4}\backslash P_4) \cup \mathcal{B}}  :f_1, f_2, f_3 \in \mathbb{F}^{ < 5}_{23}[x], f_4 \in \mathbb{F}^{ < 4}_{23}[x] \}.
     \]
     
     It can be easily verified that $(\mathcal{C}^{I_1}, \mathcal{C}^{I_2}, \mathcal{C}^{I_3}, \mathcal{C}^{I_4}; \mathcal{C}^{F})$ and the surjective map defined by $T$ and $u$ form an access-optimal MDS generalized merge-convertible code.
     Specifically, if we fix a uniformizer $\pi_{p_{\infty}}$ at place $P_{\infty}$, we see that $u_P T(f_1, f_2, f_3, f_4) (P_\infty)$ is a linear combination of $\pi_{p_{\infty}}^4 f_1(P_\infty)$, $f_2(P_{19})$, $f_3(P_{14})$ and $f_4(P_{9})$; $u_P T(f_1, f_2, f_3, f_4) (P_9)$ is a linear combination of $f_1(P_9)$, $\pi_{p_{\infty}}^4 f_2(P_{\infty})$, $f_3(P_{19})$ and $f_4(P_{14})$; $u_P T(f_1, f_2, f_3, f_4) (P_{14})$ is a linear combination of $f_1(P_{14})$, $f_2(P_{9})$, $ \pi_{p_{\infty}}^4 f_3(P_{\infty})$ and $f_4(P_{19})$; $u_P T(f_1, f_2, f_3, f_4) (P_{19})$ is a linear combination of $f_1(P_{19})$, $f_2(P_{14})$, $f_3(P_{9})$ and $\pi_{p_{\infty}}^3 f_4(P_{\infty})$.
    Hence, the write access cost is $4$, read access cost for each redundancy symbol is $4$, total real access cost equal to $16$, this is an access-optimal and per-symbol read access-optimal MDS generalized merge-convertible code.
    
\end{example}

\section{\texorpdfstring{Construction of Access-Optimal  Generalized Merge-Convertible Code Between $(r,\delta)$-LRCs }{Construction of Access-Optimal Generalized Merge-Convertible Code between (r,δ) -LRCs}}\label{sec5}

In this section, we focus on the construction of access-optimal $(t,1)_q$ $(r,\delta)$ generalized merge-convertible codes, which extend the framework introduced in Section \ref{sec4}. We assume that all initial codes have identical parameters for simplicity. We use the superscript $I$ to denote the initial code and the superscript $F$ to represent the final code throughout this section. We will later clarify that the parameters of the initial codes may differ in Remark \ref{rdelrem}.

We first present the main theorem of this section.
\begin{theorem} \label{rdelthm}
    Assume that there exists a subgroup $\mathcal{G}$ of $\mathrm{PGL}_2(q)$ of order $(r+\delta-1)l$, and $\mathcal{G}$ has a subgroup $\mathcal{H}$ of order $r+\delta -1$. Let $m$ be a positive integer satisfying $m+1 \leq \lfloor \frac{q-2((r+\delta-1)l)+3}{(r+\delta-1)l}  \rfloor $. For some positive integers $k, t, l^{'}$ satisfying $k \in [m] $, $ t \in [l] $, $(k+l^{'})(r+\delta-1) \leq q+3-2(r+\delta-1)$ and $l \leq \text{min~}\{k, l^{'} \}$, define $n_I= (k+l^{'})(r+\delta-1)$, $k_I=kr$, $n_F = (kt+l)(r+\delta-1)$ and $k_F=tk r$. Then there exists an access-optimal  $(n_I, k_I; n_F, k_F)$  $(r,\delta)$ generalized merge-convertible code with read access cost $t r l$ and write access cost $l(r+\delta-1)$.
    
\end{theorem}

\begin{proof}
    The proof is similar to Theorem \ref{mdsthm}. Let $F$ denote a rational function field over $\mathbb{F}_q$, and consider the field extension $F/F^{\mathcal{G}}$. Since $m+1 \leq \lfloor \frac{q-2((r+\delta-1)l)+3}{(r+\delta-1)l}  \rfloor $ and by \eqref{hurw},
    we can choose $k+1$ rational places $\{ P_0, P_1, P_2, \cdots, P_k \}$ in $F^{\mathcal{G}}$ that split completely in the extension $F/F^{\mathcal{G}}$, it follows that these $k+1$ rational places also split completely in the extension $F^{\mathcal{H}}/F^{\mathcal{G}}$. For $i \in [0,k]$, denote by $\{ P_{i,1}, P_{i,2}, \cdots, P_{i,l} \}$ the set of rational places in $F^{\mathcal{H}}$  lying above $P_i$, and for $ i\in[0,k]$, $ j \in [l]$ denote by $\{ P_{i,j,1}, P_{i,j,2}, \cdots, P_{i,j,r+\delta-1} \}$ the set of rational places in $F^\mathcal{G}$ lying above $P_{i,j}$. Since the Galois group acts transitively, after a suitable reordering of the places, for  $ j \in [l]$ we may denote by $\sigma_j$ the unique element in  $\mathcal{G} \backslash \mathcal{H}$ satisfying
    $\sigma_j(P_{i,1,s})=P_{i,j,s}$ for all $i \in [0,k]$ and $s \in [r+\delta-1]$. Moreover, we can select a rational place  $P_{\infty} \in F$ that does not lie above any of the $k+1$  places $\{ P_0, P_1, P_2, \cdots, P_k \}$ in $F^{\mathcal{G}}$. Since $F$ is a rational function field, there exists a function $x \in F$ with pole divisor $(x)_{\infty} = P_{\infty}$, so that $F=\mathbb{F}_q(x)$.  The intersection $   P_{\infty} \cap F^{\mathcal{H}}$ is a rational place of $F^{\mathcal{H}}$, denote it by $Q_{\infty} $. Since $F^{\mathcal{H}}$ is also  rational, there exists a function $z$ in $F^{\mathcal{H}}$ with $(z)_{\infty} = Q_{\infty}$, implying that $F^{\mathcal{H}}=\mathbb{F}_q(z)$.
    
    Let $\mathcal{B}_{0,j}=\{P_{0,j,1},P_{0,j,2},\cdots,P_{0,j,r+\delta-1 }\}$ for $j \in [l]$,  $\mathcal{A}_{i,j}=\{P_{i,j,1},P_{i,j,2}
    ,\cdots,P_{i,j,r+\delta-1 }\}$ for $i \in [k]$ and $j \in [l]$, $\mathcal{A}_j = \bigcup\limits_{1\leq i\leq k} A_{i,j}$,  $\mathcal{B} = \bigcup\limits_{1 \leq j \leq l} \mathcal{B}_{0,j} $ and $\mathcal{A}=\bigcup\limits_{1\leq j\leq t} \mathcal{A}_{j}$, we arbitrary choose $l^{'}-l$ evaluation blocks in $\mathcal{A}_{i,j}$ except from $\mathcal{A}_{1} $ as $\mathcal{B}^{'}$, and define the initial code as 
    \[
    \mathcal{C}^I=\{ \big(f(P)\big)_{P \in \mathcal{A}_1 \cup \mathcal{B} \cup \mathcal{B}^{'}} : f\in \text{ span }\{ x^iz^j  \} \text{~for~} i \in[0,r-1], j \in[0,k-1] \},
    \]
    it is clear that  $\mathcal{C}^I$ is an optimal $(n_I, k_I;(r,\delta))$-LRC. 
    
    From $t$ functions $f_1,f_2,\cdots, f_{t}$ in span $\{x^iz^j \}$, we now define the transformation map as 
    \[
    T(f_1,f_2,\cdots, f_{t})=   \sum^{t}_{j=1} g^{r-1}_{j,1} g^{k-1}_{j,2}  (h_{\mathcal{A} \backslash \mathcal{A}_j} \circ z) \sigma_j(f_j)
    \]
where $g_{j,1}$ is a fixed function with $(g_{j,1}) = \sigma_j(P_{\infty}) - P_{\infty} $, $g_{j,2}$ is a fixed function with $(g_{j,2}) = \sigma_j(\text{Con~}_{F/F^{\mathcal{H}}}(Q_{\infty})) - \text{Con~}_{F/F^{\mathcal{H}}}(Q_{\infty}) $ (see the definition in \cite[Definition 3.1.8]{stichtenoth2009algebraic}) and $(h_{\mathcal{A} \backslash \mathcal{A}_i} \circ z)(P^{'}) = \prod\limits_{\alpha \in \{ z(P) \text{~for~} P\in \mathcal{A} \backslash \mathcal{A}_i \} }(z(P^{'})-\alpha)$ for any rational place $P^{'}$ except for pole places of $z$. In particular, we choose $g_{1,1}=g_{1,2}=1$.
    Then we define the coefficient vector $ u = (u_P)_{P \in \mathcal{A}_1\cup \mathcal{A}_2 \cup  \cdots \cup  \mathcal{A}_{t} \cup \mathcal{B}}  $ as follow:
     \[
     u_{P}=
     \begin{cases}
      
      [ \frac{1}{g^{r-1}_{i,1}} \frac{1}{g^{k-1}_{i,2}} (h_{ \mathcal{A} \backslash \mathcal{A}_{i} } \circ z)^{-1} ](P) , & \text{if } P \in \mathcal{A}_{i} \text{~for~}i \in [t]; \\
      1, & \text{if } P \in \mathcal{B}.
     \end{cases}
     \]     
    then define the final code  $\mathcal{C}^F$ as:
     \[
     \mathcal{C}^F = \{ \big(u_P T(f_1, f_2, \cdots, f_{t}) (P)\big)_{P \in  \mathcal{A} \cup \mathcal{B}} : f_1, f_2, \cdots, f_{t} \in \text{span}\{ x^iz^j  \} \text{~for~} i \in[0,r-1], j \in[0,k-1] \}.
     \]
     
     In the following part, we show that this definition is well-defined.
     We need to valuate $T(f_1, f_2, \cdots, f_{t})$ at places in $ \mathcal{A} \cup \mathcal{B}$, however the only possible poles for $\sigma_j(f_j)$ are $\sigma_j(P_{\infty})$ and $\sigma_j(\text{Con~}_{F/F^{\mathcal{H}}}(Q_{\infty}))$, after multiply with $g^{r-1}_{j,1} g^{k-1}_{j,2}$, the only possible poles of $g^{r-1}_{j,1} g^{k-1}_{j,2}  (h_{\mathcal{A} \backslash \mathcal{A}_j} \circ z) \sigma_j(f_j)$ are $P_{\infty}$ and $\text{Con~}_{F/F^{\mathcal{H}}}(Q_{\infty})$.  Since $P_{\infty}$ lie over $Q_{\infty}$ and $P_{\infty}$ does not lie over ${P_0, P_1, \cdots, P_k}$, we can evaluate $T(f_1, f_2, \cdots, f_{t})$ at places in $ \mathcal{A} \cup \mathcal{B}$, the definition of $\mathcal{C}^F$ is well-defined.

     To illustrate the conversion procedure more clearly, we distinguish the initial codes in the proof. Specifically, we assume that $\mathcal{C}^{I_i}=\mathcal{C}^I=\{ (f_i(P))_{P \in \mathcal{A}_1 \cup \mathcal{B} \cup \mathcal{B}^{'}}: f_i\in \text{ span }\{ x^{i'}z^{j'}  \} \text{~for~} i' \in[0,r-1], j' \in[0,k-1] \}$ for $i \in [t]$. We now present the conversion map
 \[\phi: \mathcal{C}^{I_1}\times \mathcal{C}^{I_{2}}\times\cdots \times \mathcal{C}^{I_t} \longrightarrow \mathcal{C}^F\] defined by\[\phi\left(\big((f_1(P))_{P\in\mathcal{A}_1 \cup \mathcal{B} \cup \mathcal{B}^{'}}, \cdots,(f_t(P))_{P\in\mathcal{A}_1 \cup \mathcal{B} \cup \mathcal{B}^{'}}\big)\right)=\left(u_P T(f_1, f_2, \cdots, f_{t}) (P)\right)_{P \in  \mathcal{A} \cup \mathcal{B}}, \]
 for any $f_1,f_2,\cdots, f_{t} \in \text{ span } \{x^iz^j\} $.

     For $P_{i,j,k} \in \mathcal{A}$, 
     \begin{align*}
         u_{P_{i,j,k}}T(f_1,f_2,\cdots, f_{t}) (P_{i,j,k}) &=  [ \sum^{t}_{j^{'}=1} u_{P_{i,j^{'},k}} g^{r-1}_{j^{'},1} g^{k-1}_{j^{'},2}  (h_{\mathcal{A} \backslash \mathcal{A}_{j^{'}}} \circ z) \sigma_{j^{'}}(f_{j^{'}})] (P_{i,j,k}) \\
          &= [u_{P_{i,j,k}} g^{r-1}_{j,1} g^{k-1}_{j,2}  (h_{\mathcal{A} \backslash \mathcal{A}_j} \circ z) \sigma_j(f_j)] (P_{i,j,k}) \\
          &= [u_{P_{i,j,k}} g^{r-1}_{j,1} g^{k-1}_{j,2}  (h_{\mathcal{A} \backslash \mathcal{A}_j} \circ z)  \sigma_j(f_j)] (\sigma_j(P_{i,1,k})) \\
          &= f_j(P_{i,1,k}).
     \end{align*}
    From the above equation, $\phi$ induces identity maps from $\mathcal{C}^{I_j}|_{\mathcal{A}_1}$ to $\mathcal{C}^F|_{\mathcal{A}_j}$ for $j \in [t]$. Thus, the symbols in $\mathcal{C}^F$ corresponding to the places in $\mathcal{A}$ are the unchanged symbols, and hence the symbols corresponding to places in $\mathcal{B}$ are the written symbols. Thus, the write access cost is $|\mathcal{B}|=l(r+\delta-1)$. 
   Moreover, since $\mathcal{C}^I|_{\mathcal{A}_1}$ has dimension $kr$, and we can recover all the polynomial $f_1,f_2,\cdots, f_{t}$ from  $ u_{P_{i,j,k}}T(f_1,f_2,\cdots, f_{t}) $, hence dimension of the final code is $t kr$. Also, we have $(\sigma_j(f_j)) + (r-1)(\sigma_j(P_{\infty})) + (k-1)(\sigma_j(\text{Con}_{F/F^{\mathcal{H}}}(Q_{\infty}))) \geq 0 $ as divisor, hence $( g^{r-1}_{j,1} g^{k-1}_{j,2}  (h_{\mathcal{A} \backslash \mathcal{A}_j} \circ z) \sigma_j(f_j) ) \\
    + $ $ (kt-k) (\text{Con}_{F/F^{\mathcal{H}}}(Q_{\infty}))+ (r-1)P_{\infty} +(k-1) (\text{Con}_{F/F^{\mathcal{H}}}(Q_{\infty})) \geq 0 $ as divisor, we have $T(f_1, f_2, \cdots, f_{t}) \in \mathcal{L}(G) $ with $\text{deg}(G)= (r-1)+(kt-1)(r+\delta-1) $, it can have at most $(r-1)+(kt-1)(r+\delta-1)$ zeros, hence the minimal distance $d_F$ of the final code is at least $n_F- (r-1)-(kt-1)(r+\delta-1) = l(r+\delta-1)+\delta $. By Singleton-type bound, we have $d_F \leq (kt+l)(r+\delta-1)-t kr+1-(t k-1)(\delta-1) = (r+\delta-1)l+\delta$, hence the final code is an optimal  $(n_F,k_F; (r,\delta))$-LRC.
     
     Next, we show that this construction has read access cost $tl r $.
     For $P_{0,j,k} \in \mathcal{B}$, 
     \begin{align*}
         u_{P_{0,j,k}}T(f_1,f_2,\cdots, f_{t}) (P_{0,j,k}) &=  [ \sum^{t}_{j^{'}=1} g^{r-1}_{j^{'},1} g^{k-1}_{j^{'},2}  (h_{\mathcal{A} \backslash \mathcal{A}_{j^{'}}} \circ z) \sigma_{j'}(f_{j^{'}})] (P_{0,j,k}) \\
          &=  \sum^{t}_{j^{'}=1} [g^{r-1}_{j^{'},1} g^{k-1}_{j^{'},2}   (h_{\mathcal{A} \backslash \mathcal{A}_{j^{'}}} \circ z) ](P_{0,j,k}) [\sigma_{j'}(f_{j^{'}})] (P_{0,j,k})   \\
          &=    \sum^{t}_{j^{'}=1} [g^{r-1}_{j^{'},1} g^{k-1}_{j^{'},2}   (h_{\mathcal{A} \backslash \mathcal{A}_{j^{'}}} \circ z) ](P_{0,j,k}) f_{j^{'}} (\sigma^{-1}_{j'} (P_{0,j,k})).
     \end{align*}
    Hence $u_{P_{0,j,k}}T(f_1,f_2,\cdots, f_{t}) (P_{0,j,k})$ is a linear combination of $f_1(P_{0,j,k})$, $f_2(\sigma^{-1}_2 (P_{0,j,k}))$,$\cdots$, $f_{t}(\sigma^{-1}_t (P_{0,j,k}))$ for $j \in [l]$.
    We can recover every written symbol by $t$ symbols, and the total per-symbol read access cost equal to $t l(r+\delta-1)$.
    Moreover, due to local properties, we can recover any symbol from $r$ fixed symbols locally, hence the read access cost of this construction can be $t lr$, so it is access-optimal by Corollary \ref{rdelcor}.       
    \end{proof}

\begin{remark} \label{rdelrem}
    Each of the first three items in Remark \ref{remmds} admits a counterpart under the $(r,\delta)$ setting. Since items two and three were already illustrated with examples in the preceding chapter, we will not elaborate on them with further examples here. In the subsequent Remark \ref{remram}, we demonstrate that ramified rational places may also be used for evaluation.
\end{remark}

\begin{corollary} \label{rdelcor1}
      
Let \(\mathbb{F}_q\) be a finite field with \(q = p^s\), where \(p\) is a prime. Suppose that \(u\) is a common divisor of \(q - 1\) and \(p^v - 1\) for some \( v \in [s] \), and let \(u'\) and \(v'\) be positive integers with \(u' \mid u\), \(v' \mid v\)  and\(~ u^{'}|~p^{v^{'}}-1\). Then there exists a subgroup \(\mathcal{G}\) of \(\operatorname{PGL}_2(q)\) of order \(up^v\) which contains a subgroup \(\mathcal{H}\) of order \(u'p^{v'}\).

Let \(r\) and \(\delta\) be positive integers satisfying
\[
r + \delta - 1 = u'p^{v'},
\]
Also, let \(m\) be a positive integer such that
\[
m + 1 \le \left\lfloor \frac{q - p^v}{up^v} \right\rfloor.
\]
For some positive integers \(k\), \(t\), and \(l'\) satisfying
\[
k \in [m],\quad t \in [ \frac{u}{u'}p^{v-v'}],\quad (k + l')u'p^{v'} \le q - p^v,\quad \text{and} \quad \frac{u}{u'}p^{v-v'} \le \min\{k, l'\},
\]
define
\[
n_I = (k + l')u'p^{v'},\quad k_I = kr,\quad n_F = (kt + l')u'p^{v'},\quad \text{and} \quad k_F = tk r.
\]
Then there exists an access‐optimal \((n_I, k_I; n_F, k_F)\) \((r,\delta)\) merge-convertible code with read access cost  \(t r\, \frac{u}{u'}p^{v-v'}\) and write access cost \(up^v\).
      
\end{corollary}

\begin{proof}
    
     We only need to establish the existence of $\mathcal{G}$ and $\mathcal{H}$ with the desired orders and determine the number of available rational places by analyzing their splitting properties. Once this is done, we can apply Theorem \ref{rdelthm}.

Since \( u \mid q-1 \) (resp. \( u' \mid q-1 \)), there exists a subgroup \( H \) (resp. \( H' \)) of the multiplicative group \( \mathbb{F}_q^* \) of order \( u \) (resp. \( u' \)). Given that \( u \mid (p^v - 1) \) (resp. \( u' \mid (p^{v'} - 1) \)), the field \( \mathbb{F}_p(H) \) (resp. \( \mathbb{F}_p(H') \)) is contained in \( \mathbb{F}_{p^v} \) (resp. \( \mathbb{F}_{p^{v'}} \)).  

Define  
\[
l = \min \{ t > 0 : u \mid (p^t - 1) \} \quad (\text{resp. } l' = \min \{ t > 0 : u' \mid (p^t - 1) \}).
\]  
Then, we have \( \mathbb{F}_p(H) = \mathbb{F}_{p^l} \) (resp. \( \mathbb{F}_p(H') = \mathbb{F}_{p^{l'}} \)), and it follows that  
\[
l \mid \gcd(v, s) \quad (\text{resp. } l' \mid \gcd(v', s)).
\]  

Viewing \( \mathbb{F}_q \) as a \(\frac{s}{l}\)-dimensional vector space over \( \mathbb{F}_{p^l} \), we select a \(\frac{v}{l}\)-dimensional subspace \( W \subseteq \mathbb{F}_q \) over \( \mathbb{F}_{p^l} \). Then, the set  
\[
\mathcal{G} = \left\{ \begin{bmatrix} a & b \\ 0 & 1 \end{bmatrix} : a \in H, \, b \in W \right\}
\]  
forms a subgroup of \( \mathrm{PGL}_2(q) \) of order \( u p^v \).  

Since \( W \) is a \(\frac{v}{l}\)-dimensional vector space over \( \mathbb{F}_{p^l} \), it can also be regarded as a vector space of dimension \(\frac{v}{l'}\) over \( \mathbb{F}_{p^{l'}} \).
Thus, the subgroup  
\[
\mathcal{G}' = \left\{ \begin{bmatrix} a & b \\ 0 & 1 \end{bmatrix} : a \in H', \, b \in W \right\}
\]  
has order \( u' p^v \) and is a subgroup of \( \mathcal{G} \).  

Now, choosing a \(\frac{v'}{l'}\)-dimensional subspace \( W' \subseteq W \) over \( \mathbb{F}_{p^{l'}} \), we define  
\[
\mathcal{H} = \left\{ \begin{bmatrix} a & b \\ 0 & 1 \end{bmatrix} : a \in H', \, b \in W' \right\}.
\]  
By construction, \( \mathcal{H} \) is a subgroup of \( \mathcal{G} \) of order \( u' p^{v'} \). Hence, we find the groups $\mathcal{G}$ and $\mathcal{H}$ of the desired orders. 

Let $F$ be the rational function field $\mathbb{F}_q(x)$, consider the field extension $F/F^{\mathcal{G}}$, from Hurwitz genus formula \eqref{hurw}, \cite[Proposition IV.2]{8865641} investigates the splitting behavior of all rational places in $F$. Specifically, there exists a rational place $P_{\infty}$ of $F$ which totally ramified in $F/F^{\mathcal{G}}$, and there is a rational place of $F^{\mathcal{G}}$ which splits into $p^v$ rational places of $F$, each place has ramification index $u$, and there are $\frac{q-p^v}{up^v}$ rational places of $F^\mathcal{G}$ that split completely in $F/F^{\mathcal{G}}$. Therefore, the corollary follows from the proof of Theorem \ref{rdelthm}.    
\end{proof}

\begin{remark}
    
Take \(\delta=2\). Then, by Corollary \ref{rdelcor1}, we have constructed access-optimal LRC merge-convertible codes. Suppose we choose a subgroup \(\mathcal{G}\) of the multiplicative group \(\mathbb{F}_q^*\) of order \(u\). Since \(\mathcal{H}\) is a subgroup of \(\mathcal{G}\), it can also be viewed as a subgroup of \(\mathbb{F}_q^*\). In fact, the method for selecting the rational places corresponding to the information symbols in the access-optimal LRC merge-convertible codes in \cite[Corollary III.2]{kong2024locally} is essentially the same as ours. Consequently, the parameters of the convertible code obtained in this setting agree with those given in \cite[Corollary III.2]{kong2024locally}. However, since \cite{kong2024locally} does not construct convertible codes via automorphisms of the rational function field, the per-symbol read access cost in our codes is lower. Moreover, our construction only requires the existence of two nested subgroups in the automorphism group of a rational function field, which makes the achievable parameters in our setting more flexible than those in \cite{kong2024locally}. In particular, as we will show in Corollary \ref{rdelcor2} and Remark \ref{remram}, by using dihedral groups to construct convertible codes in even characteristic, we obtain access-optimal LRC merge-convertible codes with new parameters that have not been previously reported.

\end{remark}

In the following Corollary \ref{rdelcor2}, we construct convertible codes using dihedral groups over finite fields of even characteristic. We focus on the even characteristic case because, for odd $q$, a subgroup of $\mathrm{PGL}_2(q)$ isomorphic to a dihedral group has order $2u$, which must divide $q+1$ or $q-1$ (see \cite{cameron20063} for the proof). In such case, one can instead use a cyclic subgroup of $\mathrm{PGL}_2(q)$ to construct convertible codes with the same parameters.

\begin{corollary} \label{rdelcor2}

    Let $\mathbb{F}_q$ be a finite field of even characteristic. Suppose that $u,v$ are two positive integers satisfying  $u~|~q-1 (\text{resp. }  u~|~q+1  ) $ and $v~|~2u$. Then there exists a subgroup $\mathcal{G}$ of \(\operatorname{PGL}_2(q)\) of order \(2u\) which contains a subgroup \(\mathcal{H}\) of order \(v\).

Let \(r\) and \(\delta\) be positive integers satisfying
\[
r + \delta - 1 = v,
\]
Also, let \(m\) be a positive integer such that
\[
m + 1 \le \left\lfloor \frac{q - 1-u}{2u} \right\rfloor (\text{resp. } m + 1 \le \left\lfloor \frac{q + 1-u}{2u} \right\rfloor   ).
\]
For some positive integers \(k\), \(t\), and \(l'\) satisfying
\[
k \in  [m],\quad  t \in [ \frac{2u}{v}],\quad (k + l')v \le q - 1-u \quad (\text{resp. }  (k + l')v \le q + 1-u  )  ,\quad \text{and} \quad \frac{2u}{v} \le \min\{k, l'\},
\]
define
\[
n_I = (k + l')v,\quad k_I = kr,\quad n_F = (kt + l')v,\quad \text{and} \quad k_F = kt r.
\]
Then there exists an access‐optimal \((n_I, k_I; n_F, k_F)\) \((r,\delta)\) merge-convertible code with read access cost  \(t r\, \frac{2u}{v}\) and write access cost \(2u\).
    
\end{corollary}

\begin{proof}

    For $u ~|~ q+1$, let $F$ be the rational function field $\mathbb{F}_q(x)$, and let $f(x) = x^2 + ax + b \in \mathbb{F}_q[x]~$be a primitive polynomial of order $q^2 - 1$. Define two automorphisms of $F$ by $\eta(x)=\frac{1}{bx+a}$ and $\tau(x) = \frac{1}{bx}$. It is straightforward to verify that $\eta\tau\eta=\tau$. Hence, the group $\mathcal{G}=\langle \eta^{\frac{q+1}{u}}, \tau \rangle$ is isomorphic to a dihedral group of order $2u$. For any $v|2u$, there exists a subgroup $\mathcal{H} \subset \mathcal{G}$ of order $v$. This follows from the facts that $ \langle \eta^{\frac{q+1}{u}} \rangle$ is cyclic and that $ \tau\eta^j={\eta^{-j}}\tau $ for any positive integer $j$.
    Thus, we have identified groups $\mathcal{H}$ and $\mathcal{G}$ of the required orders. Now consider the field extension  $ F^{\langle \eta^{\frac{q+1}{u}} \rangle} / F^\mathcal{G} $. Since the extension degree $[F^{\langle \eta^{\frac{q+1}{u}} \rangle} : F^\mathcal{G} ]=2$, any ramified place in this extension must be ramified wildly. Moreover, since $\text{deg~Diff}(F^{\langle \eta^{\frac{q+1}{u}} \rangle} / F^\mathcal{G})=2 $, there is exactly one rational place of $F^\mathcal{G}$ that is ramified, and its ramification index must be $2$.
    Next, we consider the extension $F/F^{\langle \eta^{\frac{q+1}{u}} \rangle}$. As shown in Corollary \ref{poleva}, all rational places of $F$ lying above $\frac{q+1}{u}$ completely split rational places in $F^{\langle \eta^{\frac{q+1}{u}} \rangle}$. Among these $\frac{q+1}{u}$ rational places of $F^{\langle \eta^{\frac{q+1}{u}} \rangle}$, there must be one lying above the unique ramified rational place in  $F^\mathcal{G}$; otherwise, all rational places of $F$ would lie above completely split rational places in $F^\mathcal{G}$, which would imply $2u ~|~ q+1$, contradicting with $q$ even. In summary, there are $q+1-u$  rational places in $F$ that split completely in the extension $F/F^{\mathcal{G}}$, and there are $u$ rational places in $F$ that are ramified in this extension, each with ramification index $2$. The result for the case $u ~|~ q+1$ then follows from Theorem \ref{rdelthm}.

     For $u ~|~ q-1$, the proof is similar, we omit it.
    \end{proof}

\begin{remark} \label{remram}
~
    \begin{itemize}
        \item[(i)] let $\delta=2$, if we choose $\mathcal{H}$ to be a subgroup of $\mathcal{G}$ with even order in Corollary \ref{rdelcor2}, then the local group have even size. Hence, we can obtain access-optimal LRC merge-convertible codes with new parameters.
        \item[(ii)] If we take $\mathcal{H}$ to be a cyclic subgroup of order $v$ in $\mathcal{G}$ in the proof of Corollary \ref{rdelcor2}, then, by the analysis in Corollary \ref{rdelcor2}, we have $k=\frac{q+1-u}{2u}$ rational places $\{ P_1, P_2, \cdots, P_k \}$ in $F^{\mathcal{G}}$ that split completely in the extension $F/F^{\mathcal{G}}$. For $i \in [k]$, let $\{ P_{i,1}, P_{i,2}, \cdots, P_{i,l} \}$ denote the set of rational places in $F^{\mathcal{H}}$  lying above $P_i$ where $l=\frac{2u}{v}$, and let $\{ P_{i,j,1}, P_{i,j,2}, \cdots, P_{i,j,r+\delta-1} \}$ denote the set of rational places in $F^\mathcal{G}$ lying above $P_{i,j}$ where $r+\delta-1=v$. After a suitable reordering of the places, for  $ j \in [l]$ we may denote by $\sigma_j$ the unique element in  $\mathcal{G} \backslash \mathcal{H}$ satisfying
    $\sigma_j(P_{i,1,s})=P_{i,j,s}$ for all $i \in [k]$ and $s \in [r+\delta-1]$.   
    There are $u$ rational places remaining in $F$.  After a suitable reordering of the places, we can denote these $u$ rational places as $\{ P_{0,j,s} \}$ for $j\in [\frac{l}{2}],s \in [r+\delta-1]$, such that $\{P_{0,j,s}\}$ with fixed $j$ for $s \in [r+\delta-1]$ form a complete orbit under the action of $\mathcal{H}$, and it can be shown that there are exactly two elements in $\sigma_j$ for $j\in[l]$ that send $P_{0,1,s}$ to $P_{0,j,s}$ for all $s \in [r+\delta-1]$.
     Let  $\mathcal{A}_{i,j}=\{P_{i,j,1},P_{i,j,2},\cdots,P_{i,j,r+\delta-1 }\}$ for $i \in [k]$ and $j \in [l]$,
    $\mathcal{B}_{0,j}=\{P_{0,j,1},P_{0,j,2},\cdots,P_{0,j,r+\delta-1 }\}$ for $j \in [\frac{l}{2}]$,
    $\mathcal{A}_j = \bigcup\limits_{1\leq i\leq k} \mathcal{A}_{i,j}$, $\mathcal{A}=\bigcup\limits_{1\leq j\leq t} \mathcal{A}_{j}$ and $\mathcal{B}=\bigcup\limits_{1\leq j\leq \frac{l}{2}} \mathcal{B}_{0,j}$. Formally, if we take $P_{0,1,1}$ as $P_{\infty}$ in the proof of Theorem \ref{rdelthm} and use similar method in Corollary \ref{poleva}  to evaluate functions at their poles, a convertible code can be defined following the construction in Theorem \ref{rdelthm}. Such convertible code has fewer redundancy symbols, which is preferable. This is because when the number of redundancy symbols exceeds that of information symbols, the default method achieves the lower bound of access cost, as stated in Corollary \ref{rdelcor}.

    \end{itemize}
   
\end{remark}

\section{\texorpdfstring{Construction of Access-Optimal Generalized Merge-Convertible Code From MDS Codes To $(r,\delta)$-LRC}{Construction of Access-Optimal Generalized Merge-Convertible Code From MDS Codes To (r,δ) -LRC}} \label{sec6}

In this section, we present the construction of an access-optimal generalized merge-convertible code that converts MDS codes into an optimal $(r,\delta)$-LRC.

\subsection{GRS codes and generalized Vandermonde matrix}
Our construction is based on generalized Reed-Solomon (GRS) codes and the generalized Vandermonde matrix. Thus, we first briefly review some related results.

Suppose $1 \leq k <n$. Let $\alpha_1, \alpha_{2},\cdots, \alpha_n$ be $n$ distinct elements of $\mathbb{F}_q$ and $\mathbf{v}=(v_1,v_2,\cdots,v_n) \in (\mathbb{F}^*_q)^n$. 

\begin{definition}[Generalized Vandermonde matrix and GRS codes]
The generalized Vandermonde matrix  associated to $\alpha_1, \alpha_{2},\cdots, \alpha_n$ and $\mathbf{v}$ is defined as: 
\[V_k(\alpha_1, \alpha_{2},\cdots, \alpha_n; \mathbf{v}) \triangleq\left(\begin{array}{cccc}
     v_1& v_2& \cdots &v_n \\
    v_1\alpha_{1}& v_2\alpha_{2}&\cdots& v_n\alpha_{n} \\
   \vdots & \vdots & \ddots & \vdots \\
     v_1\alpha^{k-1}_{1}& v_2\alpha^{k-1}_{2} & \cdots & v_n\alpha^{k-1}_{n}\\
\end{array}\right) \in \mathbb{F}^{k \times n}_q.\]
When $\mathbf{v} = \mathbf{1}_n$, the all-one vector of length $n$, we abbreviate $V_k(\alpha_1, \alpha_2, \dots, \alpha_n; \mathbf{v})$ by $V_k(\alpha_1, \alpha_2, \dots, \alpha_n)$ for simplicity.

    The $k$-dimensional generalized Reed-Solomon (GRS) code associated to $\alpha_1, \alpha_{2},\cdots, \alpha_n$ and $\mathbf{v}$, denoted by \\ $\mathrm{GRS}_k(\alpha_1, \alpha_{2},\cdots, \alpha_n; \mathbf{v})$, is defined as the linear code with generator matrix $V_{k}(\alpha_1, \alpha_{2},\cdots, \alpha_n; \mathbf{v})$. We abbreviate  \\$\mathrm{GRS}_k(\alpha_1, \alpha_2, \dots, \alpha_n; \mathbf{1}_n)$ by $\mathrm{GRS}_k(\alpha_1, \alpha_2, \dots, \alpha_n)$ for simplicity.
    
    Each element $\alpha_i$ is called a locator of the Vandermonde matrix $V_{k}(\alpha_1, \alpha_2, \dots, \alpha_n; \mathbf{v})$ and of the GRS code \\ $\mathrm{GRS}_k(\alpha_1, \alpha_{2},\cdots, \alpha_n; \mathbf{v})$. The set  $\{\alpha_1, \alpha_{2},\cdots, \alpha_n\}$ is called the locator set. 
\end{definition}

The following result concerns the existence of a certain parity-check matrix of GRS codes, which will be used in our construction.
\begin{lemma}\label{VI.1}
    Suppose $n \geq m$ and let $\Gamma =\{\alpha_{i_1}, \alpha_{i_2},\cdots, \alpha_{i_{m}}\}\subseteq \{\alpha_1, \alpha_{2},\cdots, \alpha_n\}$, then there exists a vector $\mathbf{v}\in (\mathbb{F}^*_q)^{n}$ such that the GRS code $\mathrm{GRS}_{k}(\alpha_1, \alpha_{2},\cdots, \alpha_n; \mathbf{v})$ has a parity-check matrix $H$ satisfying $H|_{\Gamma}=V_{n-k}(\alpha_{i_1}, \alpha_{i_2},\cdots, \alpha_{i_{m}})$. 
\end{lemma}
\begin{proof}
  Let $\mathcal{C}=\mathrm{GRS}_{n-k}(\alpha_{1}, \alpha_{2},\cdots, \alpha_{n})$, which has a generator matrix $H=V_{n-k}(\alpha_{1}, \alpha_{2},\cdots, \alpha_{n})$. Then, by \cite[Proposition 5.2]{roth2006introduction}, there exists a vector $\mathbf{v} \in (\mathbb{F}_q^*)^n$, such that $\mathcal{C}^{\top}=\mathrm{GRS}_{k}(\alpha_{1}, \alpha_{2},\cdots, \alpha_{n};\mathbf{v})$. Consequently, \\$\mathrm{GRS}_{k}(\alpha_{1}, \alpha_{2},\cdots, \alpha_{n};\mathbf{v})$ has a parity-check matrix $H$. Clearly, we have $H|_{\Gamma}=V_{n-k}(\alpha_{i_1}, \alpha_{i_2},\cdots, \alpha_{i_{m}})$.
\end{proof}

\subsection{Convert MDS codes to an LRC}
To present our access-optimal convertible code from MDS codes to an LRC, we first construct $t$ initial $[n_{I_i},k_{I_i}]$-MDS codes $\mathcal{C}^{I_{i}}$ for $i \in [t]$, and a final code $\mathcal{C}^F$ which is an optimal $(n_F, k_F, d_F;(r, \delta))$-LRC. To simplify our discussion, we assume that all the initial codes have the same dimension, i.e., $k_{I_i} = k_{I}$ for $i \in [t]$, and hence $k_F = tk_{I}$. Furthermore, we assume that $d_F \leq k_{I}$ and $d_F \leq n_{I_i} - k_{I} + 1$. Otherwise, by Theorem \ref{III.2}, the read cost for each initial code is at least $k_{I}$. Thus, using the default approach, the conversion procedure becomes trivial.

Denote $\langle tk_{I} \rangle_r$ as the integer $b$ with $tk_{I} \equiv b\pmod r$ and $b \in [r]$. If $\langle tk_{I} \rangle_r \geq k_I+1$, then $\lceil \frac{tk_{I}}{r} \rceil-\lceil \frac{(t-1)k_{I}}{r} \rceil=0$ and $\lfloor \frac{k_{I}-d_F+1}{r+\delta-1}  \rfloor=0$.

Recall that $\mathcal{U}_i$ and $\mathcal{R}_i$ denote the index sets of unchanged symbols and read symbols, respectively, during the conversion procedure. By Theorem \ref{III.1},  we have
\[|\mathcal{U}_i| \leq k_{I}+  n_F-k_F-d_F+1-\left( \left\lceil \frac{k_F-k_{I}}{r} \right\rceil -1\right)(\delta-1)= k_{I}+\left( \left\lceil \frac{tk_{I}}{r} \right\rceil-\left\lceil \frac{(t-1)k_I}{r} \right\rceil\right)(\delta-1)=k_{I}.\]
By Theorem \ref{III.2}, we have
\[|\mathcal{R}_i| \geq k_{I} - (|\mathcal{U}_i\backslash \mathcal{R}_i|-d_F+1) +  (\delta-1) \left\lfloor \frac{|\mathcal{U}_i\backslash \mathcal{R}_i|-d_F+1}{r+\delta-1}  \right\rfloor \geq d_F-1+(\delta-1) \left\lfloor \frac{k_{I}-d_F+1}{r+\delta-1}  \right\rfloor=d_F-1.\]

Therefore, when $d_F\leq k_{I}$, $d_F\leq n_{I_i}-k_{I}+1$ and $\langle tk_{I} \rangle_r \geq k_{I}+1$, the convertible code is access-optimal if $|\mathcal{U}_i|=k_{I}$ and $|\mathcal{R}_i|=d_F-1$ for each $i \in [t]$.

\textbf{The parameter setting}: Suppose $\delta \geq 2$, $t=st'$ and $r=sk_{I}+a$ with $s \geq 2$, $a \geq 1$ and $at' \leq k_{I}-\delta$. We let 
\[k_F=tk_I  \textnormal{ and } n_F=t'(r+\delta-1),\]
then it can be verified that 
        \[d_F=n_F-k_F+1-\left(\left\lceil \frac{k_F}{r}\right\rceil-1\right)(\delta-1)=at'+\delta\leq k_{I},\]
 and $\langle tk_{I} \rangle_r=sk_{I}+a-at'\geq (s-1)k_{I}+a+\delta>k_{I}+1.$

\textbf{Construction of the final code $\mathcal{C}^F$}: Suppose $n_F-(t'-1)(\delta-1) \leq q$. 
Let $\alpha_{i,j}$, $\beta_{i',j'}$ and $\gamma_{\ell}$ ($i=1, \cdots, t;  j=1, \cdots, k_{I}; i'=1,\cdots,t'; j'=1,\cdots,a$; $\ell=1,\cdots, \delta-1$) be $tk_{I}+at'+\delta-1=t'(sk_{I}+a)+\delta-1=n_F-(t'-1)(\delta-1)$ distinct elements of $\mathbb{F}_q$. Let
        \[H^F=\left(\begin{array}{cccc}
    A_{1} & &  &\\
    & A_{2} &  & \\
     &  &\ddots & \\
     && &A_{t'} \\
     B_1 &B_2  &\cdots& B_{t'}
\end{array}\right),\]
where $A_{i}=V_{\delta-1}(\alpha_{s(i-1)+1,1}, \cdots, \alpha_{s(i-1)+1,k_I}, \cdots, \alpha_{si,1}, \cdots, \alpha_{si,k_I}, \beta_{i,1}, \cdots, \beta_{i,a},\gamma_1, \cdots,\gamma_{\delta-1})\in \mathbb{F}^{(\delta-1)\times (r+\delta-1)}_q$, and

\setlength{\arraycolsep}{0.9pt}
\[B_{i}=\left(\begin{array}{ccccccccccccc}
    \alpha^{\delta-1}_{s(i-1)+1,1}& \cdots& \alpha^{\delta-1}_{s(i-1)+1,k_I} & \cdots & \alpha^{\delta-1}_{si,1}& \cdots& \alpha^{\delta-1}_{si,k_I} & \beta^{\delta-1}_{i,1} & \cdots & \beta^{\delta-1}_{i,a} &\gamma^{\delta-1}_1 &\cdots &\gamma^{\delta-1}_{\delta-1}  \\
    \alpha^{\delta}_{s(i-1)+1,1}& \cdots& \alpha^{\delta}_{s(i-1)+1,k_I} & \cdots & \alpha^{\delta}_{si,1}& \cdots& \alpha^{\delta}_{si,k_I} & \beta^{\delta}_{i,1} & \cdots & \beta^{\delta}_{i,a} &\gamma^{\delta}_1 &\cdots &\gamma^{\delta}_{\delta-1}\\
   \vdots & \ddots & \vdots & \ddots & \vdots &  \ddots & \vdots &  \vdots& \ddots &  \vdots & \vdots &\ddots & \vdots \\
     \alpha^{at'+\delta-2}_{s(i-1)+1,1}& \cdots& \alpha^{at'+\delta-2}_{s(i-1)+1,k_I} & \cdots & \alpha^{at'+\delta-2}_{si,1}& \cdots& \alpha^{at'+\delta-2}_{si,k_I} & \beta^{at'+\delta-2}_{i,1} & \cdots & \beta^{at'+\delta-2}_{i,a} &\gamma^{at'+\delta-2}_1 &\cdots &\gamma^{at'+\delta-2}_{\delta-1}\\
\end{array}\right) \in \mathbb{F}^{at'\times (r+\delta-1)}_q\]
\setlength{\arraycolsep}{1.0pt}
for $i \in [t']$.

 Let $\mathcal{C}^F$ be the linear code with parity-check matrix $H^F$.
Then, we have 
\begin{theorem}
Suppose $s \geq 2$, $a \geq 1$ with $at' \leq k_I-\delta$ and $t'(sk_I+a)+\delta-1 \leq q$, then the code $\mathcal{C}^F$ constructed above is an optimal $\big(n_F=t'(sk_{I}+a+\delta-1), ~k_F=st'k_I,~ d_F=at'+\delta;~(r=sk_{I}+a,\delta) \big)$-$\mathrm{LRC}$.
\end{theorem}
\begin{proof}
    The code length, dimension, and locality of $\mathcal{C}^F$ are obvious. By the Singleton-type bound \eqref{Sbound2},
    \[d(\mathcal{C}^F) \leq n_F-k_F+1-\left(\left\lceil \frac{k_F}{r}\right\rceil-1\right)(\delta-1)=at'+\delta.\]
    So we only need to show that any submatrix $M$ of $H_F$, consisting of any $at'+\delta-1$ columns of $H^F$ has full column rank.  Suppose the columns of $M$ come from $m$ repair groups. If one group contains fewer than $\delta-1$ columns, then these columns are linearly independent of the others, since any $\delta-1$ columns of $A_{i}$ are linearly independent. Therefore, we can assume that each of the $m$ groups contains at least $\delta-1$ columns of $M$.  Note that the locator sets of $A_1, A_2, \cdots, A_{t'}$ share a common subset $\{\gamma_1, \cdots, \gamma_{\delta-1}\}$. 
    Thus, we can rearrange the columns of $M$ and delete the zero rows to obtain the following matrix:
    \[M_1=\left(\begin{array}{cccc|cccc}
    A'_{1} & & & &A''_{1} & &  &\\
    & A'_{2} &  && & A''_{2} &  &\\
     &  &\ddots &  & & &\ddots &\\
     && &A'_{m} &&& &A''_{m}\\
     B'_1 &B'_2  &\cdots& B'_{m}&B''_1 &B''_2  &\cdots& B''_{m}
\end{array}\right),\]
   where each $A'_{i}$ is a $(\delta-1)\times(\delta-1)$ Vandermonde matrix and hence is invertible, and $A''_{1}, \cdots, A''_{m}$ are Vandermonde matrices with disjoint locator sets. We use $A'_i$ to eliminate $A''_i$ via column operations, yielding the following matrix:
   \[M_2=\left(\begin{array}{c|c|c|c|c|c|c|c}
    A'_{1} & & & &\mathbf{0} & &  &\\
    & A'_{2} &  && & \mathbf{0} &  &\\
     &  &\ddots &  & & &\ddots &\\
     && &A'_{m} &&& &\mathbf{0}\\
     B'_1 &B'_2  &\cdots& B'_{m}&B''_1-B'_1(A'_{1})^{-1}A''_{1} &B''_2-B'_2(A'_{2})^{-1}A''_{2}  &\cdots& B''_{m}-B'_m(A'_{m})^{-1}A''_{m}
\end{array}\right)\]
Then we only need to show that the matrix
\[\left(\begin{array}{c|c|c|c}
    B''_1-B'_1(A'_{1})^{-1}A''_{1} &B''_2-B'_2(A'_{2})^{-1}A''_{2}  &\cdots& B''_{m}-B'_m(A'_{m})^{-1}A''_{m}
\end{array}\right)\]
is of full column rank. Consider the matrix
\[N=\left(\begin{array}{cccc|cccc}
    A'_{1} &  A'_{2}& \cdots &A'_{m} &A''_{1}&A''_{2} 
     &\cdots &A''_{m}\\
     B'_1 &B'_2  &\cdots& B'_{m}&B''_1 &B''_2  &\cdots& B''_{m}
\end{array}\right).\]
Note that $\left(\begin{array}{c}
    A'_{i}\\
     B'_i
\end{array}\right)$, $i=1,2,\cdots, m$  are Vandermonde matrices. By removing their repeated locators, we obtain a new matrix:
\[N_1=\left(\begin{array}{c|cccc}
    U &A''_{1}&A''_{2} 
     &\cdots &A''_{m}\\
    V&B''_1 &B''_2  &\cdots& B''_{m}
\end{array}\right),\]
where $\left(\begin{array}{c}
    U\\
    V
\end{array}\right)$, $\left(\begin{array}{c}
    A''_{i}\\
     B''_i
\end{array}\right)$, $i=1,2,\cdots, m$ are Vandermonde matrices with distinct locators. Then $N_1$ is a Vandermonde matrix. Note that $N_1$ has $at'+\delta-1$ rows and at most $at'+\delta-1$ columns. Thus $N_1$ is of full column rank. By the construction, $\left(\begin{array}{c}
    A'_{i}\\
     B'_i
\end{array}\right)$ is still a submatrix of $\left(\begin{array}{c}
    U\\
    V
\end{array}\right)$. So we can use the same column operations to eliminate the matrix $A''_i$ and obtain a new matrix 
\[N_2=\left(\begin{array}{c|c|c|c|c}
    U &\mathbf{0}&\mathbf{0} 
     &\cdots &\mathbf{0}\\
    V& B''_1-B'_1(A'_{1})^{-1}A''_{1} &B''_2-B'_2(A'_{2})^{-1}A''_{2}  &\cdots& B''_{m}-B'_m(A'_{m})^{-1}A''_{m}
\end{array}\right),\]
which is of full column rank.
We then deduce that $\left(\begin{array}{c|c|c|c}
    B''_1-B'_1(A'_{1})^{-1}A''_{1} &B''_2-B'_2(A'_{2})^{-1}A''_{2}  &\cdots& B''_{m}-B'_m(A'_{m})^{-1}A''_{m}
\end{array}\right)$ 
is of full column rank. The proof is completed.
\end{proof}
\begin{remark}
    It is worth noting that our construction is inspired by the work of \cite{jin2019explicit} and \cite{xing2018construction}; however, it differs from theirs, as it does not satisfy the condition stated in \cite[Theorem III.3]{xing2018construction}.
\end{remark}

We use the element $\alpha_{i,j}$ ($i\in [t],~ j \in [k_{I}]$) as the index of the corresponding column of $H^F$. Let $\mathcal{U}_{i}=\{\alpha_{i,1}, \cdots, \alpha_{i,k_{I}}\}$, which represents the unchanged symbols of our final code during the conversion procedure. The remaining columns of $H^F$ are indexed by $\mathcal{W}$, which correspond to the newly written symbols of our final code in the conversion procedure.

Note that $H^F|_{\mathcal{U}_i}$ has $at'+\delta-1=d_F-1$ nonzero rows. Let $\mathcal{T}_i$ denote the set of row indices corresponding to its zero rows. We use $\mathbf{0}_{\mathcal{T}_i \times k_{I}}$ to denote these zero rows, and let $\widehat{H}_i \in \mathbb{F}^{(d_F-1) \times k_I}_q$ denote the submatrix of $H^F|_{\mathcal{U}_i}$ consisting of its nonzero rows. Then 
\[\widehat{H}_i=V_{at'+\delta-1}(\alpha_{i,1}, \cdots, \alpha_{i,k_I}).\]
For the convenience of our construction, we set \[H^F|_{\mathcal{U}_i}=\left(\begin{array}{cc}
     \widehat{H}_i  \\
     \mathbf{0}_{\mathcal{T}_i \times k_I} 
\end{array}\right).\] However, it should be noted that this notation is not rigorous and is only used for convenience. In fact, the zero entries are not necessarily located in the lower part of the matrix. Specifically, the rows indexed by $\mathcal{T}_i$ are zero rows, while the remaining part corresponds to $\widehat{H}_i$.

\textbf{Construction of the $t=st'$ intial MDS codes}: For any $i \in [t]$, suppose $  n_{I_i} \in  [k_{I}+at'+\delta-1 , q]$, and let \\
${\alpha_{i,1}, \cdots, \alpha_{i,k_{I}}, \xi_{i,1}, \cdots, \xi_{i,n_{I_i}-k_{I}}}$ be distinct elements in $ \mathbb{F}_q$. By Lemma~\ref{VI.1}, there exists a vector  \\
$\mathbf{w}_i=(w_{i,1}, \cdots, w_{i,k_{I}+at'+\delta-1}) \in (\mathbb{F}_q^*)^{k_{I}+at'+\delta-1}$ such that the GRS code $\mathrm{GRS}_{k_{I}}(\alpha_{i,1}, \cdots, \alpha_{i,k_{I}}, \xi_{i,1}, \cdots, \xi_{i,at'+\delta-1}; \mathbf{w}_i)$ has a parity-check matrix $\overline{H}_i \in \mathbb{F}_q^{(at'+\delta-1) \times (k_{I}+at'+\delta-1)}$ satisfying $\overline{H}_i|_{\mathcal{U}_i} = V_{at'+\delta-1}(\alpha_{i,1}, \cdots, \alpha_{i,k_{I}})$.

Then we define 
\[\mathcal{C}^{I_{i}}=\mathrm{GRS}_{k_{I}}(\alpha_{i,1}, \cdots, \alpha_{i,k_{I}},  \xi_{i,1}, \cdots, \xi_{i,n_{I_i}-k_{I}}; (\mathbf{w}_i,1,\cdots,1)).\] 
$\mathcal{C}^{I_{i}}$ is an $[n_{I_i},k_{I}]$-MDS code.
Let $\mathcal{R}_i=\{\xi_{i,1}, \cdots, \xi_{i,at'+\delta-1}\}$, which are the read symbols in the conversion procedure. Then, the puncture code $\mathcal{C}^{I_{i}}|_{\mathcal{U}_i \cup \mathcal{R}_i}=\mathrm{GRS}_{k_{I}}(\alpha_{i,1}, \cdots, \alpha_{i,k_{I}}, \xi_{i,1}, \cdots, \xi_{i,at'+\delta-1}; \mathbf{w}_i)$ and $\overline{H}_i$ is its parity-check matrix, which satisfies that 
\[\overline{H}_i|_{\mathcal{U}_i}=V_{at'+\delta-1}(\alpha_{i,1}, \cdots, \alpha_{i,k_{I}})=\widehat{H}_i.\]

We now present the main theorem of this section.
\begin{theorem}\label{thmVI.2}
Suppose $s \geq 2$, $a \geq 1$ with $at' \leq k_{I}-\delta$, $t'(sk_{I}+a)+\delta-1 \leq q$, and $k_{I}+at'+\delta-1 \leq n_{I_i} \leq q$.
    The initial $[n_{I_i},k_{I}]$-$\mathrm{MDS}$ codes $\mathcal{C}^{I_{i}}$  ($i\in [t])$ and the final optimal $\big(n_F=t'(sk_{I}+a+\delta-1), ~k_F=st'k_{I},~ d_F=at'+\delta; ~(r=sk_{I}+a,\delta) \big)$-$\mathrm{LRC}$ $\mathcal{C}^F$ constructed above form an access-optimal generalized merge-convertible code from $\mathrm{MDS}$ codes to an optimal $(r, \delta)$-$\mathrm{LRC}$.
\end{theorem}
\begin{proof}
To prove the conclusion, we first need to give the conversion map
     \[\phi: \mathcal{C}^{I_1} \times \cdots \times \mathcal{C}^{I_{t}} \longrightarrow \mathcal{C}^{F}\]
   such that $\phi(\mathbf{c}_1,\mathbf{c}_2,\cdots,\mathbf{c}_{t})=\mathbf{c}_F=(\mathbf{c}_F|_{\mathcal{U}_1},\cdots, \mathbf{c}_F|_{\mathcal{U}_t}, \mathbf{c}_F|_{\mathcal{W}}) \in \mathcal{C}^{F}$ with $\mathbf{c}_i \in \mathcal{C}^{I_i}$ and 
$\mathbf{c}_F|_{\mathcal{U}_i}=\mathbf{c}_i|_{\mathcal{U}_i}$. Note that $|\mathcal{R}_i|=at'+\delta-1=d_F-1$. Thus, the only thing we need to show is that the written symbols $\mathbf{c}_F|_{\mathcal{W}}$ can be generated by the read symbols $\mathbf{c}_i|_{\mathcal{R}_i}$, $i \in [t]$.
     From $\mathbf{0}=\mathbf{c}_F\cdot (H^F)^\top=(\mathbf{c}_F|_{\mathcal{U}_1},\cdots,\mathbf{c}_F|_{\mathcal{U}_t},\mathbf{c}_F|_{\mathcal{W}})\cdot(H^F|_{\mathcal{U}_1},\cdots,H^F|_{\mathcal{U}_t}, H^F|_{\mathcal{W}})^\top$, we have
     \begin{eqnarray*}
      \mathbf{c}_F|_{\mathcal{W}}\cdot (H^F|_{\mathcal{W}})^\top&=& -\sum_{i=1}^t\mathbf{c}_F|_{\mathcal{U}_i}\cdot(H^F|_{\mathcal{U}_i})^\top\\
      &=& -\sum_{i=1}^t\mathbf{c}_i|_{\mathcal{U}_i}\cdot(H^F|_{\mathcal{U}_i})^\top\\
      &=&-\sum_{i=1}^t\mathbf{c}_i|_{\mathcal{U}_i}\cdot\left(\begin{array}{cc}
     \widehat{H}_i  \\
     \mathbf{0}_{\mathcal{T}_i \times k_{I}}
\end{array}\right)^\top\\
      &=&-\sum_{i=1}^t\big(\mathbf{c}_i|_{\mathcal{U}_i}\cdot(\overline{H}_i|_{\mathcal{U}_i})^\top, \mathbf{c}_i|_{\mathcal{U}_i}\cdot\mathbf{0}_{\mathcal{T}_i \times k_{I}}^\top\big)\\
      &=&\sum_{i=1}^t\big(\mathbf{c}_i|_{\mathcal{R}_i}\cdot(\overline{H}_i|_{\mathcal{R}_i})^\top, \mathbf{0}_{\mathcal{T}_i}\big),
     \end{eqnarray*}
     where the first term of the last equality holds since $\mathbf{c}_i|_{\mathcal{U}_i}\cdot(\overline{H}_i|_{\mathcal{U}_i})^\top+\mathbf{c}_i|_{\mathcal{R}_i}\cdot(\overline{H}_i|_{\mathcal{R}_i})^\top=\mathbf{c}_i|_{\mathcal{U}_i\cup \mathcal{R}_i}\cdot(\overline{H}_i)^\top=\mathbf{0},$ while the second term of the last equality holds because the columns of $\mathbf{0}_{\mathcal{T}_i \times k_{I}}^\top$  are indexed by $\mathcal{T}_i$ and we use $\mathbf{0}_{\mathcal{T}_i}$ to denote the zero row vector indexed by $\mathcal{T}_i$.
     Finally, we only need to show that the $t'(a+\delta-1)\times t'(a+\delta-1)$ matrix $H^F|_{\mathcal{W}}$ is invertible. Actually, by the construction of $H^F$, we have 
     \[H^F|_{\mathcal{W}}=\left(\begin{array}{c|c|c|c}
    P_{1} & &  &\\
    & P_{2} &  & \\
     &  &\ddots & \\
     && &P_{t'} \\
     Q_1 &Q_2  &\cdots& Q_{t'}
\end{array}\right),\]
where $P_{i}=V_{\delta-1}(\beta_{i,1},\cdots, \beta_{i,a}, \gamma_1,\cdots, \gamma_{\delta-1})$ and 
\[Q_{i}=\left(\begin{array}{cccccc}
 \beta^{\delta-1}_{i,1} & \cdots & \beta^{\delta-1}_{i,a} &\gamma^{\delta-1}_1 &\cdots &\gamma^{\delta-1}_{\delta-1}\\
     \beta^{\delta}_{i,1} & \cdots & \beta^{\delta}_{i,a} &\gamma^{\delta}_1 &\cdots &\gamma^{\delta}_{\delta-1}  \\
     \vdots& \ddots &  \vdots & \vdots &\ddots & \vdots \\
      \beta^{at'+\delta-2}_{i,1} & \cdots & \beta^{at'+\delta-2}_{i,a} &\gamma^{at'+\delta-2}_1 &\cdots &\gamma^{at'+\delta-2}_{\delta-1}\\
\end{array}\right).\]
Then  $H^F|_{\mathcal{W}}$ is  equivalent to    \[H_1=\left(\begin{array}{c|c|c|c}
    P_{1} & &  &\\
    & P_{2} &  & \\
     &  &\ddots & \\
     P_{1}&P_{2}& \cdots&P_{t'} \\
     Q_1 &Q_2  &\cdots& Q_{t'}
\end{array}\right).\]
Denote $M_{i}=\left(\begin{array}{c}
    P_{i} \\
     Q_{i}\\
\end{array}\right)$, then 
\[H_1=\left(\begin{array}{c|c|c|c|c}
    P_{1} & &  & &\\
    & P_{2} &  & &\\
     &  &\ddots & &\\
     & & &P_{t'-1} & \\
     M_1 &M_2  &\cdots& M_{t'-1}&M_{t'}
\end{array}\right).\]
Note that $M_{i}=V_{at'+\delta-1}(\beta_{i,1},\cdots, \beta_{i,a}, \gamma_1, \cdots, \gamma_{\delta-1})=\big(V_{at'+\delta-1}(\beta_{i,1},\cdots, \beta_{i,a}),  V_{at'+\delta-1}(\gamma_1, \cdots, \gamma_{\delta-1})\big)$. For convenience, we denote $(M'_{i}=V_{at'+\delta-1}(\beta_{i,1},\cdots, \beta_{i,a})$, $T=V_{at'+\delta-1}(\gamma_1, \cdots, \gamma_{\delta-1})$, $P'_i=V_{\delta-1}(\beta_{i,1},\cdots, \beta_{i,a})$ and $P=V_{\delta-1}(\gamma_1,\cdots, \gamma_{\delta-1})$. Then
$M_{i}=\big(M'_{i},  T\big)$, $P_{i}=(P'_{i}, P)$ and
\[H_1=\left(\begin{array}{cc|cc|c|cc|cc}
    P'_{1}&P & & && & & &\\
    && P'_{2}&P & & & & &\\
     & && &\ddots & & &\\
      & &&&  &P'_{t'-1}&P &&\\
     M'_{1}&T&M'_{2}&T& \cdots&M'_{t'-1}&T &M'_{t'}&T\\
\end{array}\right).\]
We use the rightmost $T$ to eliminate all the other $T$ by performing column operations on $H_1$, resulting in the following matrix:
\[H_2=\left(\begin{array}{cc|cc|c|cc|cc}
    P'_{1}&P & & && & & &\\
    && P'_{2}&P & & & & &\\
     & && &\ddots & & &\\
      & &&&  &P'_{t'-1}&P &&\\
     M'_{1}&\mathbf{0}&M'_{2}&\mathbf{0}& \cdots&M'_{t'-1}&\mathbf{0}&M'_{t'}&T\\
\end{array}\right).\]
Then 
\[\det(H_2)=\big(\det(P)\big)^{t'-1}\cdot \det(M'_1|M'_2|\cdots|M'_{t'-1}|M_{t'}|T).\]
Note that $P=V_{\delta-1}(\gamma_1,\cdots, \gamma_{\delta-1})$ is a $(\delta-1)\times(\delta-1) $ Vandermonde matrix of full rank and $(M'_1|M'_2|\cdots|M'_{t'-1}|M_{t'}|T)=V_{at'+\delta-1}(\beta_{1,1},\cdots, \beta_{1,a}, \cdots, \beta_{t',1},\cdots, \beta_{t',a},\gamma_1, \cdots, \gamma_{\delta-1})$ is a $(at'+\delta-1)\times(at'+\delta-1) $ Vandermonde matrix of full rank. Thus $H_2$ is invertible, which implies that $H^F|_{\mathcal{W}}$ is also invertible. Now we complete the proof.
\end{proof}

\begin{remark}
To the best of our knowledge, our construction is the first access-optimal generalized merge-convertible code that converts MDS codes into an optimal LRC. Since the initial MDS code we construct can be viewed as an optimal LRC with locality $k_I$, and the final code has locality $sk_I+a$, this construction also represents the first access-optimal generalized locally repairable convertible code, in which the initial and final codes have different localities.
\end{remark}
\section{Conclusion}\label{sec7}
In this paper, we focus exclusively on constructing access-optimal generalized merge-convertible codes. We begin by establishing a new lower bound on the access cost when the final code is an $(r,\delta)$-LRC, which subsumes all previously known bounds in the merge conversion setting. We then propose a construction of access-optimal and per-symbol read access-optimal convertible codes from MDS codes to an MDS code by using the automorphism group of a rational function field. This approach is further extended to construct access-optimal convertible codes from $(r,\delta)$-LRCs to an $(r,\delta)$-LRC. Finally, we provide the first explicit construction of convertible codes that enable merge conversion from MDS codes to an $(r,\delta)$-LRC.

It is worth noting that this work focuses exclusively on generalized merge-convertible codes. Investigating lower bounds on access cost when the initial code is an $(r,\delta)$-LRC in the splitting regime, as well as developing explicit constructions of convertible codes that achieve these bounds, remains an interesting avenue for future research.



\bibliographystyle{IEEEtran}
\bibliography{ref}

\end{document}